\documentclass[11pt]{article}

\usepackage[margin=1in]{geometry}
\usepackage[round]{natbib}
\usepackage{amsmath,amssymb,amsfonts,amsthm,bm}
\usepackage{graphicx}
\usepackage{booktabs,longtable}
\usepackage{multirow}
\usepackage{array}
\usepackage{subcaption}
\usepackage{xcolor}
\usepackage{bbm}
\usepackage{algorithm}
\usepackage{algorithmic}
\usepackage[hidelinks]{hyperref}

\newcommand{\argmin}{\operatorname*{arg\,min}}

\newtheorem{assumption}{Assumption}
\newtheorem{theorem}{Theorem}

\newtheorem{lemma}{Lemma}

\title{Calibration-Free Robust Bayesian Inference for Discrete Unnormalized Models}
\author{
Jiongran Wang\\Texas A\&M University
\and
Debdeep Pati\\University of Wisconsin--Madison
\and
Anirban Bhattacharya\\Texas A\&M University
}
\date{}

\begin{document}
\maketitle

\begin{abstract}
Bayesian inference for discrete models with intractable normalizing constants is computationally challenging, particularly when the sample space is large or countably infinite. Although some likelihood-based Bayesian methods can circumvent the intractable normalizing constant, their uncertainty quantification may be poorly calibrated under model misspecification. Generalized Bayesian approaches can address this calibration issue, but typically require careful tuning of a learning rate. In this paper, we propose DFD-BETEL, a semiparametric Bayesian framework that uses the first-order optimality conditions of discrete Fisher divergence as moment conditions in Bayesian exponentially tilted empirical likelihood. The resulting posterior avoids evaluation of the normalizing constant and requires no learning-rate calibration. We establish consistency and asymptotic normality of the discrete Fisher divergence estimator and prove a Bernstein–von Mises theorem for the DFD-BETEL posterior, showing that its limiting covariance matches the sampling covariance of the estimator and yields asymptotically calibrated uncertainty quantification. Numerical studies show that DFD-BETEL remains competitive with likelihood-based Bayesian inference under correct specification and provides more reliable uncertainty quantification under model misspecification. An application to automobile-insurance claim data further illustrates the practical utility of DFD-BETEL.
\end{abstract}

\section{Introduction}

Discrete probabilistic models often assign an unnormalized mass $q(\mathbf{x};\boldsymbol{\theta})$ to each configuration $\mathbf{x}\in\mathcal{X}$, with the corresponding probability mass function given by $p(\mathbf{x};\boldsymbol{\theta}) = q(\mathbf{x};\boldsymbol{\theta})/Z(\boldsymbol{\theta})$. When the discrete sample space is large or countably infinite, evaluating the parameter-dependent normalizing constant
$Z(\boldsymbol{\theta}) =\sum_{\mathbf{x}\in\mathcal{X}}q(\mathbf{x};\boldsymbol{\theta})$ can be computationally prohibitive. This difficulty arises in a variety of discrete models, including lattice and network models \citep{gutmann2013estimation,park2018bayesian}. In the Bayesian setting, such models are commonly referred to as \textit{doubly-intractable}, since both the likelihood and the resulting posterior involve intractable normalizing constants. This poses a challenge for standard posterior computation. For example, a Metropolis--Hastings (MH) update requires ratios $Z(\boldsymbol{\theta}_1)/Z(\boldsymbol{\theta}_2)$, which are
unavailable when the model normalizing constant cannot be evaluated. 

\paragraph{Related literature.} A substantial literature has developed computational strategies for doubly-intractable models. Auxiliary-variable MCMC methods introduce additional variables to eliminate ratios of normalizing constants \citep{moller2006efficient,murray2006mcmc,liang2010double, rao2016data,lee2026delayed}, while other approaches construct estimators or approximations of otherwise intractable likelihood quantities \citep{andrieu2009pseudo,lyne2015russian,yang2025correlated}. Simulation-based approaches, including approximate Bayesian computation and Bayesian synthetic likelihood, provide another route when synthetic data can be generated from the model \citep{marin2012approximate,price2018bayesian}. However, exact simulation can itself be computationally challenging for many discrete models. More recently, contrastive Bayesian inference has been developed using noise-contrastive estimation \citep{sonobe2026contrastive}; see also \citet{park2018bayesian} for a review of Bayesian computation for doubly-intractable models. While these approaches provide different ways to circumvent intractable normalizing constants, likelihood-based Bayesian uncertainty quantification can still be poorly calibrated when the working model is misspecified \citep{kleijn2012bernstein}. This raises a complementary question: can Bayesian inference for discrete unnormalized models avoid the normalizing constant while retaining reliable uncertainty quantification under model misspecification?

\citet{hyvarinen2005estimation} proposed score matching as an alternative route to inference in unnormalized models without evaluating the normalizing constant. Score matching has since become a key tool in score-based generative modeling \citep{song2019generative,lai2025principles,ranganath2026unified} and has also been connected to Bayesian inference through model selection and robust Bayesian approaches \citep{shao2019bayesian,jewson2022general,
altamirano2023robust,matsubara2024generalized,
bharti2026amortised,wang2026robust}. For discrete data, the score-matching principle was first extended to binary variables through ratio matching \citep{hyvarinen2007some}, with subsequent developments extending to more general discrete sample spaces \citep{lyu2012interpretation,xu2022generalized,
matsubara2024generalized,laplante2025conjugate}. In particular, \citet{matsubara2024generalized} developed \textit{discrete Fisher divergence} (DFD) for discrete models whose sample spaces may be multidimensional and countably infinite, and incorporated DFD as a loss in generalized Bayesian inference. While this construction avoids evaluation of the normalizing constant, the resulting posterior requires calibration of a learning rate that directly affects posterior uncertainty.

A calibration-free route to robust inference is provided by exponentially tilted empirical likelihood (ETEL; \citealp{schennach2005bayesian}), a variant of empirical likelihood (EL; \citealp{owen1988empirical,qin1994empirical,owen2001empirical}) designed for inference from moment conditions without requiring full model specification. Both EL and ETEL admit Bayesian formulations \citep{lazar2003bayesian,schennach2005bayesian}. \citet{chib2018bayesian} established the frequentist properties of Bayesian ETEL, showing that it can provide calibrated posterior uncertainty under model misspecification. Related extensions consider semiparametric inference under partial model specification \citep{tang2022bayesian}.

\paragraph{Contributions.} We develop \emph{DFD-BETEL}, a semiparametric Bayesian framework that incorporates the first-order optimality conditions of DFD as moment conditions in Bayesian ETEL. The resulting posterior preserves the computational advantage of DFD by relying only on ratios of neighboring unnormalized masses and requires no learning-rate calibration. We establish consistency and asymptotic normality of the sample DFD estimator and prove a Bernstein--von Mises (BvM) theorem for the DFD-BETEL posterior. Its limiting covariance matches the sampling covariance of the DFD estimator, yielding asymptotically calibrated credible sets under model misspecification. We further evaluate DFD-BETEL on a range of discrete intractable models and an automobile-insurance application. The numerical studies demonstrate
competitive estimation under correct specification and more reliable uncertainty quantification under misspecification. A sensitivity analysis further illustrates the dependence of generalized Bayesian uncertainty quantification on the choice of learning rate, which DFD-BETEL avoids.

% The remainder of the paper is organized as follows. Section~\ref{sec:preliminary} introduces DFD and develops the
% DFD-BETEL framework, while Section~\ref{sec:theory} establishes its large-sample properties. Section~\ref{sec:experiments} presents numerical studies under correct specification and misspecification, a learning-rate sensitivity analysis, and an application to automobile-insurance claim data. The appendices contain proofs of the theoretical results, posterior computation details, additional simulation studies under correct specification and misspecification, further learning-rate sensitivity results, and additional details for the real-data analysis.

\section{Methodology}
\label{sec:preliminary}

\paragraph{Problem Formulation.}
Let $\mathbf{x}_1,\ldots,\mathbf{x}_n$ be independent observations from an unknown distribution with probability mass function $p_{\mathbf{x}}(\cdot)$ on a discrete product space $\mathcal{X}=\mathcal{X}_1\times\cdots\times\mathcal{X}_m$, where each $\mathcal{X}_j$ is order-isomorphic to a subset of $\mathbb{Z}$. We consider a parametric working model
\begin{equation*}
p(\mathbf{x};\boldsymbol{\theta}) = \frac{q(\mathbf{x};\boldsymbol{\theta})}{Z(\boldsymbol{\theta})},
\qquad
\boldsymbol{\theta}\in\Theta\subseteq\mathbb{R}^p,
\end{equation*}
where $q(\mathbf{x};\boldsymbol{\theta})$ is analytically tractable, but the normalizing constant $Z(\boldsymbol{\theta}) :=\sum_{\mathbf{x}\in\mathcal{X}} q(\mathbf{x};\boldsymbol{\theta})$ is computationally intractable. The intractability of $Z(\boldsymbol{\theta})$ poses a major challenge for standard parameter estimation procedures, such as maximum likelihood estimation \citep{moller2006efficient}. In this work, we say that the model is \textit{correctly specified} if there exists $\boldsymbol{\theta}_0\in\Theta$ such that $p_{\mathbf{x}}(\cdot)=p(\cdot;\boldsymbol{\theta}_0)$; otherwise, the model is \textit{misspecified}.

\paragraph{Discrete Fisher Divergence.}
Discrete extensions of score matching \citep{hyvarinen2005estimation} provide an alternative route to inference without evaluating $Z(\boldsymbol{\theta})$. \citet{hyvarinen2007some} introduced ratio matching for binary variables, with subsequent developments extending this principle to more general discrete sample spaces \citep{lyu2012interpretation,xu2022generalized, matsubara2024generalized,laplante2025conjugate}. In this work, we focus on DFD \citep{matsubara2024generalized}, which accommodates multidimensional and potentially countably infinite discrete sample spaces. Under DFD, the target parameter $\boldsymbol{\theta}^*$ is defined as the minimizer of the loss $J(\cdot)$:
\begin{equation}
\label{eq:target_parameter}
\begin{aligned}
    \boldsymbol{\theta}^*
    &= \argmin_{\boldsymbol{\theta} \in \Theta} J(\boldsymbol{\theta}), \\
    J(\boldsymbol{\theta})
    &= \mathbb{E}_{p_{\mathbf{x}}}\left[
    \left\|
    \frac{\nabla^{-}p(\mathbf{x};\boldsymbol{\theta})}
         {p(\mathbf{x};\boldsymbol{\theta})}
    -
    \frac{\nabla^{-}p_{\mathbf{x}}(\mathbf{x})}
         {p_{\mathbf{x}}(\mathbf{x})}
    \right\|^2
    \right],
\end{aligned}
\end{equation}
with $|| \cdot ||$ denoting the Euclidean norm and $\mathbb{E}_{p_{\mathbf{x}}}$ denoting expectation with respect to $p_{\mathbf{x}}$. Here, $\nabla^{-} h(\mathbf{x}) := [h(\mathbf{x}) - h(\mathbf{x}^{1-}), \ldots, h(\mathbf{x}) - h(\mathbf{x}^{m-})]^\top \in \mathbb{R}^m$ denotes the backward difference operator for $h : \mathcal{X} \to \mathbb{R}$, where $\mathbf{x}^{j-}$ is obtained by replacing the $j$th coordinate with its predecessor. For a count coordinate, the predecessor decreases the coordinate by one, with a predecessor below zero assigned function value zero, while the successor $\mathbf{x}^{j+}$ increases it by one. For a finite coordinate, the predecessor and successor are defined cyclically and are inverse operations. The normalizing constant cancels out in $J(\boldsymbol{\theta})$, but the objective still involves the unknown probability mass function $p_{\mathbf{x}}(\mathbf{x})$. Under suitable regularity conditions \citep{matsubara2024generalized}, $J(\boldsymbol{\theta})$ in \eqref{eq:target_parameter} can be written, up to an additive constant independent of $\boldsymbol{\theta}$, as
\begin{equation}
\label{eq:dfd_computable}
J(\boldsymbol{\theta})
=
\mathbb{E}_{p_{\mathbf{x}}}
\left[
\sum_{j=1}^{m}
\left\{
\left(
\frac{p(\mathbf{x}^{j-};\boldsymbol{\theta})}
     {p(\mathbf{x};\boldsymbol{\theta})}
\right)^2
-
2
\left(
\frac{p(\mathbf{x};\boldsymbol{\theta})}
     {p(\mathbf{x}^{j+};\boldsymbol{\theta})}
\right)
\right\}
\right].
\end{equation}
Since the omitted term does not depend on $\boldsymbol{\theta}$, it does not affect the minimizer, and we henceforth use $J(\boldsymbol{\theta})$ to denote the computable criterion in \eqref{eq:dfd_computable}. Given $\mathbf{x}_{1:n} = (\mathbf{x}_1, \ldots, \mathbf{x}_n)$,
the sample analogue of $J(\boldsymbol{\theta})$ is
\begin{equation}
\label{eq:sample_loss}
\begin{aligned}
\widehat{J}(\boldsymbol{\theta};\mathbf{x}_{1:n})
&=
\frac{1}{n}\sum_{i=1}^n \sum_{j=1}^m
\Bigg\{
\left(
\frac{p(\mathbf{x}_i^{j-};\boldsymbol{\theta})}
     {p(\mathbf{x}_i;\boldsymbol{\theta})}
\right)^2 \\
&\qquad\qquad
-2
\left(
\frac{p(\mathbf{x}_i;\boldsymbol{\theta})}
     {p(\mathbf{x}_i^{j+};\boldsymbol{\theta})}
\right)
\Bigg\}.
\end{aligned}
\end{equation}
We denote the parameter value minimizing $\widehat{J}(\boldsymbol{\theta};\mathbf{x}_{1:n})$ by $\widehat{\boldsymbol{\theta}}_n$ and refer to it as the \textit{sample DFD estimator}.

\subsection{DFD-BETEL}

The target parameter $\boldsymbol{\theta}^*$ in \eqref{eq:target_parameter} is defined as the minimizer of the population criterion $J(\boldsymbol{\theta})$. To use this criterion within ETEL, we work with the first-order optimality condition of this minimization problem. If $\boldsymbol{\theta}^*$ is an interior minimizer and differentiation and expectation can be interchanged, $\boldsymbol{\theta}^*$ satisfies the population estimating equation
\begin{equation}
\label{eq:population_moment_condition}
\mathbb{E}_{p_{\mathbf{x}}}
\left\{
\mathbf{g}(\mathbf{x}, \boldsymbol{\theta}^*)
\right\}
=
\mathbf{0}_p,
\end{equation}
where $\mathbf{0}_p$ is the $p\times 1$ vector of zeros. Using the computable representation in \eqref{eq:dfd_computable}, define
\begin{equation}
\label{eq:dfd_moment}
\mathbf{g}(\mathbf{x},\boldsymbol{\theta})
=
\nabla_{\boldsymbol{\theta}}
\left[
\sum_{j=1}^m
\left\{
\left(
\frac{p(\mathbf{x}^{j-};\boldsymbol{\theta})}
     {p(\mathbf{x};\boldsymbol{\theta})}
\right)^2
-
2
\frac{p(\mathbf{x};\boldsymbol{\theta})}
     {p(\mathbf{x}^{j+};\boldsymbol{\theta})}
\right\}
\right],
\end{equation}
where $\nabla_{\boldsymbol{\theta}}$ denotes the gradient with respect to $\boldsymbol{\theta}$. As established above, the ratios appearing in \eqref{eq:dfd_moment} are free of the normalizing constant $Z(\boldsymbol{\theta})$. Motivated by this estimating equation, we propose \textbf{DFD-BETEL}, which incorporates the first-order optimality condition of the DFD criterion as moment restrictions in Bayesian ETEL. This construction combines the ability of DFD to circumvent intractable normalizing constants in discrete models with Bayesian ETEL to provide uncertainty quantification without learning-rate calibration. Equation \eqref{eq:population_moment_condition} defines a population moment condition. Given independent observations $\mathbf{x}_{1:n}$, we construct the corresponding sample moment functions $\mathbf{g}(\mathbf{x}_i,\boldsymbol{\theta})$, for $i=1,\ldots,n$. The ETEL function $L:\Theta\to[0,\infty)$ is constructed from probability weights obtained by solving a constrained optimization problem \citep{schennach2005bayesian}. The optimization problem is feasible when the origin lies in the interior of the convex hull $\mathcal{C}(\boldsymbol{\theta}):=\{\sum_{i=1}^n\omega_i\mathbf{g}(\mathbf{x}_i,\boldsymbol{\theta})
:\omega_i\geq0,\ \sum_{i=1}^n\omega_i=1\}$ of $\{\mathbf{g}(\mathbf{x}_i,\boldsymbol{\theta})\}_{i=1}^n$. Recent work on regularized ETEL provides a way to relax this convex-hull restriction \citep{kim2023regularized}. Here, we focus on the standard ETEL framework, for which the ETEL function is defined as
\begin{equation}
\label{eq:ETEL}
L(\boldsymbol{\theta})
=
\mathbbm{1}\left\{
\mathbf{0}_p\in
\operatorname{int}\mathcal{C}(\boldsymbol{\theta})
\right\}
\prod_{i=1}^n \omega_i(\boldsymbol{\theta}),
\end{equation}
where $\mathbbm{1}\{\cdot\}$ denotes the indicator function, and
$\boldsymbol{\omega}(\boldsymbol{\theta}) := \big(\omega_1(\boldsymbol{\theta}),\ldots, \omega_n(\boldsymbol{\theta})\big)$ solves
\begin{equation*}
\label{eq:optimization_problem}
\begin{aligned}
\max_{\boldsymbol{\omega} \in \mathcal{S}^{n-1}}
\quad &
\sum_{i=1}^n -\omega_i \log(n\omega_i) \\
\text{subject to}\quad &
\sum_{i=1}^n
\omega_i\,\mathbf{g}(\mathbf{x}_i,\boldsymbol{\theta})
=
\mathbf{0}_p,
\end{aligned}
\end{equation*}
where
$\mathcal{S}^{n-1} := \{v \in \mathbb{R}^n : v_i \geq 0,\ i=1,\ldots,n,\ \sum_{i=1}^n v_i=1\}$ denotes the $(n-1)$-dimensional unit simplex. Introducing Lagrange multipliers for the constraints, the weights $\boldsymbol{\omega}(\boldsymbol{\theta})$ admit the representation \citep{schennach2005bayesian}
\begin{equation*}
\begin{aligned}
\omega_i(\boldsymbol{\theta})
&=
\frac{
\exp\!\left(
\widehat{\boldsymbol{\lambda}}(\boldsymbol{\theta})^\top
\mathbf{g}(\mathbf{x}_i,\boldsymbol{\theta})
\right)}
{
\sum_{j=1}^n
\exp\!\left(
\widehat{\boldsymbol{\lambda}}(\boldsymbol{\theta})^\top
\mathbf{g}(\mathbf{x}_j,\boldsymbol{\theta})
\right)
}, \\
\widehat{\boldsymbol{\lambda}}(\boldsymbol{\theta})
&=
\argmin_{\boldsymbol{\lambda}\in\mathbb{R}^p}
\frac{1}{n}
\sum_{i=1}^n
\exp\!\left(
\boldsymbol{\lambda}^\top
\mathbf{g}(\mathbf{x}_i,\boldsymbol{\theta})
\right).
\end{aligned}
\end{equation*}
The weights $\boldsymbol{\omega}(\boldsymbol{\theta})$ can be interpreted as probability weights assigned to the sample observations that are closest to the empirical weights $(n^{-1},\ldots,n^{-1})$ in Kullback--Leibler (KL) divergence, subject to the weighted sample moment condition. Using $L(\boldsymbol{\theta})$ as a surrogate likelihood, we define
the DFD-BETEL posterior density as
\begin{equation}
\label{eq:dfd_betel_posterior}
\pi(\boldsymbol{\theta} \mid \mathbf{x}_{1:n})
=
\frac{
L(\boldsymbol{\theta})\pi(\boldsymbol{\theta})
}{
\int_{\Theta}
L(\boldsymbol{\eta})\pi(\boldsymbol{\eta})
\,d\boldsymbol{\eta}
},
\end{equation}
where $\pi(\boldsymbol{\theta})$ denotes the prior density. Posterior sampling can be carried out using a MH algorithm. Since $L(\boldsymbol{\theta})=0$ whenever the origin does not lie in the interior of $\mathcal{C}(\boldsymbol{\theta})$, any proposal violating this convex-hull condition is rejected. Algorithm~\ref{alg:mh_dfd_betel} in Appendix~\ref{app:algorithm} summarizes the MH sampler for the DFD-BETEL posterior. Having established the DFD-BETEL construction, we next study its large-sample properties.

\section{Theoretical Properties}
\label{sec:theory}

In this section, we study the large-sample properties of DFD-BETEL. Under an i.i.d. setting with fixed parameter dimension $p$, we establish consistency and asymptotic normality of the sample DFD estimator $\widehat{\boldsymbol{\theta}}_n$ and a BvM theorem for the DFD-BETEL posterior in \eqref{eq:dfd_betel_posterior}. Taken together, these results characterize both the sampling covariance of the DFD estimator and the asymptotic behavior of the DFD-BETEL posterior, providing a theoretical basis for its uncertainty quantification. We first introduce the notation and assumptions used throughout this section.

Let $\varphi(\mathbf{x},\boldsymbol{\theta})$ denote the expression inside the expectation in \eqref{eq:dfd_computable}, and let $\mathbf{g}(\mathbf{x},\boldsymbol{\theta}) =\nabla_{\boldsymbol{\theta}}\varphi(\mathbf{x},\boldsymbol{\theta})$ denote the corresponding estimating function in \eqref{eq:dfd_moment}. Define $\boldsymbol{\Delta}
=\mathbb{E}_{p_{\mathbf{x}}}
[\mathbf{g}(\mathbf{x},\boldsymbol{\theta}^*)
\mathbf{g}(\mathbf{x},\boldsymbol{\theta}^*)^\top]$ and
$\boldsymbol{\Gamma}
=\mathbb{E}_{p_{\mathbf{x}}}
[\partial\mathbf{g}(\mathbf{x},\boldsymbol{\theta}^*)/
\partial\boldsymbol{\theta}^\top]$. We also write $\ell_n(\boldsymbol{\theta})=\log L(\boldsymbol{\theta})$
for the log-ETEL function corresponding to \eqref{eq:ETEL}, and $\mathcal{N}(A;\boldsymbol{\mu},\boldsymbol{\Sigma})$ for $\Pr(\mathbf{Z}\in A)$ when $\mathbf{Z}\sim N(\boldsymbol{\mu},\boldsymbol{\Sigma})$. We impose the following assumptions:

\begin{assumption}[Identification]
\label{ass:identification}
The parameter space $\Theta\subset\mathbb{R}^p$ is compact, and the
population DFD criterion $J(\boldsymbol{\theta})$ has a unique
minimizer $\boldsymbol{\theta}^*\in\operatorname{int}(\Theta)$.
Moreover, $\varphi(\mathbf{x},\boldsymbol{\theta})$ is continuous in
$\boldsymbol{\theta}$ for each $\mathbf{x}\in\mathcal{X}$ and
$\mathbb{E}_{p_{\mathbf{x}}}
[\sup_{\boldsymbol{\theta}\in\Theta}
|\varphi(\mathbf{x},\boldsymbol{\theta})|]<\infty$.
\end{assumption}

\begin{assumption}[Local regularity]
\label{ass:regularity}
There exists an open neighborhood $\mathcal{N}$ of
$\boldsymbol{\theta}^*$ such that
$\varphi(\mathbf{x},\boldsymbol{\theta})$ is twice continuously
differentiable in $\boldsymbol{\theta}$ on $\mathcal{N}$ for each
$\mathbf{x}\in\mathcal{X}$,
$\mathbb{E}_{p_{\mathbf{x}}}
[\sup_{\boldsymbol{\theta}\in\mathcal{N}}
\|\mathbf{g}(\mathbf{x},\boldsymbol{\theta})\|^\alpha]<\infty$
for some $\alpha>2$, and
$\mathbb{E}_{p_{\mathbf{x}}}
[\sup_{\boldsymbol{\theta}\in\mathcal{N}}
\|\nabla_{\boldsymbol{\theta}}
\mathbf{g}(\mathbf{x},\boldsymbol{\theta})\|_F]<\infty$, where $\|\cdot\|_F$ denotes the Frobenius norm.
\end{assumption}

\begin{assumption}[Nonsingularity]
\label{ass:nonsingularity}
The matrix $\boldsymbol{\Delta}$ is nonsingular and
$\boldsymbol{\Gamma}$ has full rank $p$.
\end{assumption}

\begin{assumption}[Prior and ETEL regularity]
\label{ass:posterior}
The prior admits a density $\pi(\boldsymbol{\theta})$ that is
continuous and positive at $\boldsymbol{\theta}^*$. For every
$\delta>0$, there exists $\epsilon>0$ such that, as $n\to\infty$,
\[
\mathbb{P}\left\{
\sup_{\|\boldsymbol{\theta}-\boldsymbol{\theta}^*\|>\delta}
\frac{
\ell_n(\boldsymbol{\theta})
-\ell_n(\boldsymbol{\theta}^*)
}{n}
\leq -\epsilon
\right\}
\to1.
\]
Moreover, there exists a neighborhood $\mathcal{N}_0$ of
$\boldsymbol{\theta}^*$ such that, with probability tending to one, $\widehat{\boldsymbol{\lambda}}(\boldsymbol{\theta})$ is well-defined and unique and $\mathbf{0}_p\in\operatorname{int}\mathcal{C}(\boldsymbol{\theta})$ for every $\boldsymbol{\theta}\in\mathcal{N}_0$.
\end{assumption}

Assumptions~\ref{ass:identification}--\ref{ass:nonsingularity} provide identification, smoothness, moment, and nonsingularity conditions for the DFD estimating equations, analogous to those used in \citet{newey2004higher,schennach2007point,chib2018bayesian}. Assumption~\ref{ass:posterior} imposes regularity on the prior, separation of the log-ETEL criterion, and local feasibility of the ETEL construction. Together, these conditions provide the regularity needed for the Bayesian ETEL asymptotic argument of \citet{chib2018bayesian}. We next state the BvM theorem.

\begin{theorem}[Bernstein--von Mises theorem for DFD-BETEL]
\label{thm:bvm}
Under Assumptions~\ref{ass:identification}--\ref{ass:posterior}, the DFD-BETEL posterior defined in \eqref{eq:dfd_betel_posterior} satisfies
\begin{equation}
\label{eq:bvm}
\sup_{B}
\left|
\pi\left(
\sqrt{n}
(\boldsymbol{\theta}-\widehat{\boldsymbol{\theta}}_n)
\in B
\mid \mathbf{x}_{1:n}
\right)
-
\mathcal{N}(B;\mathbf{0}_p,\boldsymbol{\Sigma})
\right|
\overset{p}{\longrightarrow}0,
\end{equation}
where the supremum is taken over all Borel sets $B\subseteq\mathbb{R}^p$,
\[
\boldsymbol{\Sigma}
=
(\boldsymbol{\Gamma}^{\top}
\boldsymbol{\Delta}^{-1}
\boldsymbol{\Gamma})^{-1}
=
\boldsymbol{\Gamma}^{-1}
\boldsymbol{\Delta}
\boldsymbol{\Gamma}^{-\top},
\]
and $\widehat{\boldsymbol{\theta}}_n$ is the sample DFD estimator. The convergence in \eqref{eq:bvm} is in probability under the data-generating distribution.
\end{theorem}

The above theorem implies that the total variation distance between the DFD-BETEL posterior $\pi(\boldsymbol{\theta}\mid\mathbf{x}_{1:n})$ and a Gaussian distribution centered at the sample DFD estimator $\widehat{\boldsymbol{\theta}}_n$ with covariance
$n^{-1}\boldsymbol{\Sigma} =n^{-1}(\boldsymbol{\Gamma}^{\top}
\boldsymbol{\Delta}^{-1}\boldsymbol{\Gamma})^{-1}$ converges in probability to zero. To assess whether the posterior uncertainty reflects the sampling variability of the sample DFD estimator, we compare their limiting covariance matrices. Theorem~\ref{thm:dfd_asymptotic_normality} in Appendix~\ref{app:asymptotic_normality} shows that
\begin{equation}
\label{eq:dfd_asymptotic_normality}
\sqrt{n}
(\widehat{\boldsymbol{\theta}}_n-\boldsymbol{\theta}^*)
\overset{d}{\longrightarrow}
N\left(
\mathbf{0}_p,
\boldsymbol{\Gamma}^{-1}
\boldsymbol{\Delta}
\boldsymbol{\Gamma}^{-\top}
\right),
\end{equation}
where $\overset{d}{\longrightarrow}$ denotes convergence in
distribution. Comparing \eqref{eq:bvm} and \eqref{eq:dfd_asymptotic_normality}, the limiting covariance of the DFD-BETEL posterior coincides with the asymptotic sampling covariance of the sample DFD estimator. Consequently, nominal $100 \times (1-\alpha)\%$ DFD-BETEL credible sets are asymptotically valid $100 \times (1-\alpha)\%$ frequentist confidence sets for the target parameter $\boldsymbol{\theta}^*$. Importantly, this covariance matching holds whether the working model is correctly specified or misspecified. To see why, suppose the working model is misspecified. Then $\boldsymbol{\theta}^*$ represents the DFD pseudo-true parameter rather than a parameter of the data-generating distribution. Since $\boldsymbol{\theta}^*$ is an interior minimizer of the population DFD criterion, it still satisfies the moment condition $\mathbb{E}_{p_{\mathbf{x}}}
[\mathbf{g}(\mathbf{X},\boldsymbol{\theta}^*)]=\mathbf{0}_p$. Thus, DFD-BETEL retains asymptotically calibrated uncertainty for $\boldsymbol{\theta}^*$ even when the working model is misspecified. Moreover, unlike generalized Bayesian inference based directly on the DFD loss, which we refer to as \textit{DFD-Bayes}, DFD-BETEL achieves this calibration without introducing a learning-rate parameter.

\paragraph{Comparison with DFD-Bayes and Standard Bayes.}
We next compare the uncertainty quantification of DFD-BETEL with DFD-Bayes and standard Bayesian inference under model misspecification. We first contrast their asymptotic posterior variances and then use a simple discrete example to illustrate the resulting differences in frequentist coverage.

DFD-Bayes constructs a generalized posterior based on the DFD loss,
\[
\pi_{\eta}(\boldsymbol{\theta}\mid\mathbf{x}_{1:n})
\propto
\pi(\boldsymbol{\theta})
\exp\{-\eta n\widehat{J}_n(\boldsymbol{\theta})\},
\]
where $\eta>0$ is a constant learning rate to be specified \citep{bissiri2016general,matsubara2024generalized}. The learning rate controls the concentration of the generalized posterior and therefore requires calibration for uncertainty quantification \citep{matsubara2024generalized}. In the one-dimensional case, where $\Delta$ and $\Gamma$ denote the scalar counterparts of $\boldsymbol{\Delta}$ and $\boldsymbol{\Gamma}$, the asymptotic sampling variance of the sample DFD estimator is $\Delta/(n\Gamma^2)$, whereas the DFD-Bayes posterior has variance $(\eta\Gamma)^{-1}/n$ under standard regularity conditions. Thus, matching the posterior variance to the sampling variance of the DFD estimator requires an appropriate choice of $\eta$. In contrast, DFD-BETEL recovers the sandwich covariance through the ETEL construction without introducing a learning-rate parameter.

We next illustrate these differences, together with the behavior of standard Bayesian inference, through a simple misspecified discrete example. Suppose the data-generating distribution is geometric with success probability $1/2$, so that $p_{\mathbf{x}}(x)=2^{-(x+1)}$ for $x\in\mathbb{N}_0$, and consider the Poisson working model $p(x;\theta)=e^{-\theta}\theta^x/x!$. We take $\Theta=[10^{-6},10]$ and assign a
$\operatorname{Gamma}(1,1)$ prior truncated to $\Theta$, using the shape-rate parameterization. The working model is misspecified since the geometric distribution does not belong to the Poisson family.

For the Poisson working model, the corresponding DFD loss yields the pseudo-true parameter $\theta^*_{\mathrm{DFD}}=3/2$. At this target, $\Gamma=32/27$ and $\Delta=128/9$, so the limiting variance of $\sqrt{n}(\widehat{\theta}_n-\theta^*_{\mathrm{DFD}})$ is $\Delta/\Gamma^2=81/8$, corresponding to an asymptotic sampling variance $81/(8n)$ for $\widehat{\theta}_n$. By Theorem~\ref{thm:bvm}, the DFD-BETEL posterior has the same first-order variance $81/(8n)$, and hence its nominal 95\% credible interval has asymptotic frequentist coverage 95\% for $\theta^*_{\mathrm{DFD}}$.

For DFD-Bayes, the first-order posterior variance is $27/(32\eta n)$. This variance matches the asymptotic sampling variance of the sample DFD estimator when $\eta=1/12$. If instead $\eta=1$, the posterior is substantially more concentrated, and its nominal 95\% credible interval has asymptotic frequentist coverage
\[
\Pr\left(|Z|<\frac{1.96}{\sqrt{12}}\right)
\approx 0.43,
\qquad Z\sim N(0,1).
\]
Thus, the uncertainty quantification of DFD-Bayes depends critically on the choice of learning rate, whereas DFD-BETEL achieves covariance matching without learning-rate calibration.

The usual Bayesian posterior provides a complementary comparison. Under model misspecification, its pseudo-true parameter $\theta^\dagger$ is defined as the minimizer of the KL divergence from the data-generating distribution to the working model. In this example, $\theta^\dagger=1$. The standard Bayes posterior has variance $1/n$, whereas the sampling variance of the maximum likelihood estimator is $2/n$ \citep{white1982maximum,ghosal2017fundamentals}. Consequently, targeting its own pseudo-true parameter $\theta^\dagger$, a nominal 95\% Bayesian credible interval substantially undercovers, with asymptotic frequentist coverage
\[
\Pr\left(|Z|<\frac{1.96}{\sqrt{2}}\right)
\approx 0.83,
\qquad Z\sim N(0,1).
\]

Under correct specification, by contrast, the corresponding posterior and sampling covariances agree asymptotically, so likelihood-based Bayesian inference and appropriately calibrated DFD-based procedures yield nominal frequentist coverage under their respective regularity conditions. DFD-based inference may nevertheless be less statistically efficient than likelihood-based inference because it does not exploit the information contained in the normalizing constant; related efficiency properties of score matching have been studied by \citet{koehler2022statistical}. In the following section, we complement these theoretical calibration results with simulation studies under both correct specification and model misspecification.

\section{Experiments}
\label{sec:experiments}

We evaluate DFD-BETEL on discrete doubly-intractable models under correct specification and misspecification, focusing on estimation accuracy and uncertainty quantification. We compare DFD-BETEL with DFD-Bayes \citep{matsubara2024generalized}, double Metropolis--Hastings (DMH) \citep{liang2010double}, and, when feasible, standard Bayesian inference. DMH provides a likelihood-based benchmark for doubly-intractable models, while standard Bayes is included when the normalizing constant can be accurately approximated. We consider a correctly specified Poisson graphical model, a misspecified Potts model, a sensitivity analysis of DFD-Bayes to the learning rate, and an automobile-insurance application using the Conway--Maxwell--Poisson (CMP) model. Additional simulation results under correct specification and misspecification are provided in Appendices~\ref{app:correct_specification} and~\ref{app:pg2_misspecified}, respectively, with further learning-rate sensitivity results in Appendix~\ref{app:sensitivity}.

\paragraph{Simulation setup.}
Each simulation setting uses 500 independently generated datasets, with all competing methods applied to the same datasets. Posterior means are used as point estimates. Under correct specification, we evaluate estimation accuracy using bias, mean absolute deviation (MAD), and mean squared error (MSE). Under misspecification, we assess uncertainty quantification using empirical coverage probabilities of nominal 95\% credible intervals, with coverage evaluated relative to each method's corresponding pseudo-true parameter. Model- and method-specific computational details are provided in the Appendix~\ref{app:computation}.

\subsection{Correctly Specified Poisson Graphical Model}
\label{sec:pg2}

\begin{table*}[t]
\caption{Correctly specified Poisson graphical model: bias, MAD, and
MSE of posterior mean estimates over 500 replications. Bias and MAD
are multiplied by $10^2$, and MSE by $10^3$.}
\label{tab:pg2_correct_main}
\centering\small
\setlength{\tabcolsep}{4pt}
\begin{tabular}{rlrrrrrrrrr}
\toprule
& & \multicolumn{3}{c}{$a_1$}
& \multicolumn{3}{c}{$a_2$}
& \multicolumn{3}{c}{$\gamma_{12}$} \\
\cmidrule(lr){3-5}\cmidrule(lr){6-8}\cmidrule(lr){9-11}
$n$ & Method
& Bias & MAD & MSE
& Bias & MAD & MSE
& Bias & MAD & MSE \\
\midrule
200 & DFD-BETEL & -0.89 & 6.65 & 7.11
                 & -0.73 & 5.01 & 4.06
                 & -0.12 & 0.91 & 0.13 \\
    & DFD-Bayes & -0.20 & 6.75 & 7.21
                 & -0.20 & 5.06 & 4.11
                 & -0.06 & 0.92 & 0.13 \\
    & DMH       & -1.53 & 5.92 & 5.61
                 & -1.19 & 4.54 & 3.36
                 & -0.19 & 0.77 & 0.10 \\
\midrule
500 & DFD-BETEL & -0.53 & 4.38 & 3.08
                 & -0.33 & 3.14 & 1.58
                 & -0.07 & 0.58 & 0.05 \\
    & DFD-Bayes & -0.16 & 4.40 & 3.11
                 & -0.05 & 3.16 & 1.59
                 & -0.03 & 0.58 & 0.05 \\
    & DMH       & -0.66 & 3.79 & 2.27
                 & -0.45 & 2.71 & 1.18
                 & -0.08 & 0.48 & 0.04 \\
\midrule
1000 & DFD-BETEL & -0.08 & 3.17 & 1.62
                  & -0.01 & 2.33 & 0.89
                  & -0.01 & 0.41 & 0.03 \\
     & DFD-Bayes & 0.16 & 3.18 & 1.62
                  & 0.16 & 2.32 & 0.89
                  & 0.02 & 0.41 & 0.03 \\
     & DMH       & 0.01 & 2.70 & 1.21
                  & 0.06 & 2.04 & 0.68
                  & 0.01 & 0.34 & 0.02 \\
\bottomrule
\end{tabular}
\end{table*}

Consider a random vector $\mathbf{x}=(x_1,\ldots,x_m)$ taking values in $\mathbb{N}_0^m$, and let $G=(V,E)$ denote an undirected graph with node set $V:=\{1,\ldots,m\}$ indexing the variables $\{x_j\}_{j=1}^m$. In an exponential-family graphical model \citep{yang2015graphical}, each node-conditional distribution $p(x_j\mid \mathbf{x}_{-j})$ belongs to a univariate exponential family, while the graph encodes the conditional dependence structure among the variables.

We consider the Poisson graphical model, for which the node-conditional distributions are Poisson. To provide a correctly specified benchmark, we consider the two-node model on $\mathbb{N}_0^2$ with parameter $\boldsymbol{\theta}=(a_1,a_2,\gamma_{12})$ and probability mass function
\begin{equation}
\label{eq:pg2_working}
p(\mathbf{x};\boldsymbol{\theta}) = \frac{\exp\{a_1x_1+a_2x_2-\gamma_{12}x_1x_2\}}{Z(\boldsymbol{\theta})x_1!x_2!}.
\end{equation}

The conditional distribution of $x_j$ given $x_{3-j}$ is Poisson with mean $\exp(a_j-\gamma_{12}x_{3-j})$, for $j=1,2$. The interaction parameter $\gamma_{12}>0$ induces a negative conditional dependence between the two counts and ensures normalizability of the joint distribution. We generate data from $\boldsymbol{\theta}_0=(\log 8,\log 10,0.05)$ with sample sizes $n=200,500,1000$. We assign independent $N(0,2.5^2)$ priors to $a_1$ and $a_2$, and a half-$N(0,1)$ prior to $\gamma_{12}$. Additional implementation details are provided in Appendix~\ref{app:pg2_correct_details}.

Table~\ref{tab:pg2_correct_main} summarizes the estimation results. All three methods provide accurate estimation under correct specification, with estimation errors decreasing as the sample size increases. DFD-BETEL and DFD-Bayes exhibit similar performance, but have larger MAD and MSE than DMH across the three parameters, reflecting the expected efficiency advantage of likelihood-based inference when the model is correctly specified. As we show in the following misspecified experiment, the advantage of DFD-BETEL becomes more apparent when robustness of uncertainty quantification is required.

\subsection{Misspecified Potts Model}
\label{sec:potts}

We next consider the Potts model \citep{potts1952some}, a multistate generalization of the Ising model \citep{ising1925beitrag} in which each lattice site can take more than two possible states. For an $M\times N$ lattice, let $\mathbf{x}=(x_1,\ldots,x_{MN})$, where $x_i\in\{1,\ldots,q\}$, and let $E$ denote the set of horizontal and vertical nearest-neighbor edges. We generate data from
\begin{equation}
\label{eq:potts_true}
f(\mathbf{x}) \propto \exp\left\{\alpha_0\sum_{i=1}^{MN}\mathbbm{1}\{x_i=1\} + \beta_0\sum_{(i,j)\in E}\mathbbm{1}\{x_i=x_j\}\right\},
\end{equation}
where $\alpha_0$ controls the preference for the first state and $\beta_0$ controls the tendency of neighboring sites to occupy the same state.

We set $q=3$, use a $10\times10$ lattice, and generate data with $(\alpha_0,\beta_0)=(0.4,0.5)$. Inference is based on the interaction-only working model
\begin{equation}
\label{eq:potts_working}
p(\mathbf{x};\theta) \propto \exp\left\{\theta\sum_{(i,j)\in E}\mathbbm{1}\{x_i=x_j\}\right\},
\end{equation}
so that omission of the external-field term induces model misspecification. We consider sample sizes $n=50,70,$ and $150$ and assign $\theta\sim\operatorname{Unif}(0,2)$. The DFD and KL pseudo-true parameters are defined as the minimizers of the population DFD and KL divergence, respectively. DFD-BETEL and DFD-Bayes coverage is evaluated with respect to the DFD target, while DMH coverage is evaluated with respect to the KL target.

\begin{figure*}[t]
    \centering
    \includegraphics[width=0.92\textwidth]
        {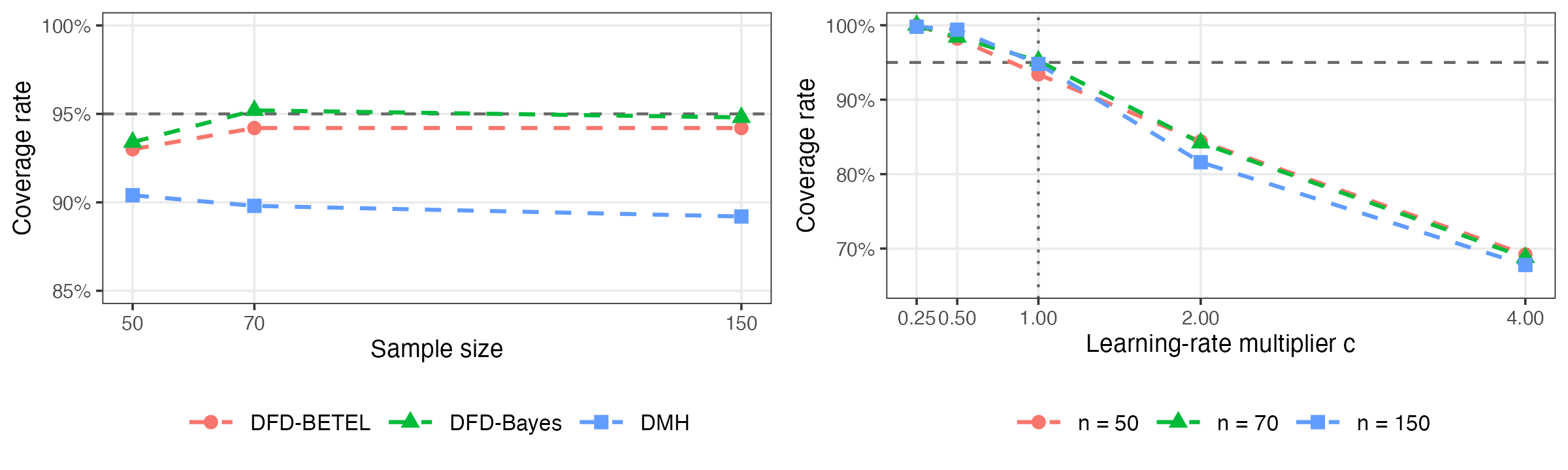}
    \caption{Empirical coverage probabilities of nominal 95\% credible intervals. Left: coverage of DFD-BETEL, DFD-Bayes, and DMH under the misspecified Potts model. Right: sensitivity of DFD-Bayes coverage to the learning-rate multiplier $c$, where $c=1$ corresponds to the calibrated learning rate. Horizontal dashed lines indicate the nominal 95\% coverage level.}
    \label{fig:potts_coverage_sensitivity}
\end{figure*}

The left panel of Figure~\ref{fig:potts_coverage_sensitivity} shows that DFD-BETEL maintains coverage close to the nominal 95\% level across all sample sizes. DFD-Bayes achieves similar coverage after learning-rate calibration, whereas DMH consistently undercovers its corresponding KL target. These results illustrate the uncertainty-calibration advantage of DFD-based inference under model misspecification. Unlike DFD-Bayes, however, DFD-BETEL achieves this calibration without requiring the choice of a learning rate. We examine the sensitivity of DFD-Bayes to this choice in the next experiment.

\subsection{Sensitivity to Learning-Rate Calibration in DFD-Bayes}
\label{sec:learning_rate}

The preceding experiment shows that DFD-Bayes can provide well-calibrated uncertainty with an appropriately chosen learning rate, but its uncertainty quantification depends on this choice. We therefore examine its sensitivity to the choice of learning rate.

For each replication, we first obtain a baseline learning rate $\widehat{\eta}$ using the calibration procedure of \citet{matsubara2024generalized} based on 100 nonparametric bootstrap samples. We then construct DFD-Bayes posteriors using learning rates $c\widehat{\eta}$ for $c\in\{0.25,0.5,1,2,4\}$, where $c=1$ corresponds to the calibrated learning rate. Within each replication, the same dataset and prior are used for all choices of $c$. We illustrate the sensitivity using the misspecified Potts model from the preceding experiment; additional sensitivity results are provided in Appendix~\ref{app:sensitivity}.

The right panel of Figure~\ref{fig:potts_coverage_sensitivity} shows that the uncertainty quantification of DFD-Bayes is sensitive to the choice of learning rate. Decreasing the learning rate produces wider, conservative credible intervals, whereas increasing it leads to substantial undercoverage. For example, at $n=150$, coverage decreases from $0.998$ at $c=0.25$ to $0.948$ at the calibrated value $c=1$ and $0.678$ at $c=4$. Although the bootstrap calibration procedure of \citet{matsubara2024generalized} performs favorably relative to alternative calibration approaches in their finite-sample experiments, they note that the bootstrap criterion for selecting the learning rate can fail in some settings, in which case the resulting generalized posterior may not be well calibrated. Similar sensitivity is observed in the additional settings reported in Appendix~\ref{app:sensitivity}. These results highlight the practical importance of learning-rate calibration for generalized Bayesian uncertainty quantification. In contrast, DFD-BETEL does not introduce a learning-rate parameter.

\subsection{Automobile Insurance Claim Data}
\label{sec:real_data}

We next apply DFD-BETEL to automobile-insurance claim data using the CMP model \citep{shmueli2005useful}, a flexible count model on $\mathbb{N}_0=\{0,1,2,\ldots\}$ with probability mass function
\begin{equation}
\label{eq:cmp_model}
p(x;\lambda,\nu) = \frac{1}{Z(\lambda,\nu)}\frac{\lambda^x}{(x!)^\nu},
\end{equation}
where $\lambda>0$ controls intensity and $\nu>0$ controls dispersion. The Poisson distribution is recovered when $\nu=1$, while $\nu>1$ and $\nu<1$ correspond to underdispersion and overdispersion, respectively. The normalizing constant $Z(\lambda,\nu)$ is an infinite series, making the CMP model a natural example of a discrete unnormalized model.

We analyze the automobile-insurance claim data of \citet{lemaire2012bonus}, which record the number of claims for $n=106{,}974$ policyholders from a Belgian automobile insurer. The sample mean and variance are $0.1011$ and $0.1074$, respectively, suggesting mild overdispersion relative to the Poisson model. We assign independent $\chi_3^2$ priors to $\lambda$ and $\nu$. Additional computational details and convergence diagnostics are provided in Appendix~\ref{app:cmp_insurance}.

The posterior means of $\lambda$ and $\nu$ are $0.0952$ and $0.3105$, with 95\% credible intervals $[0.0933,0.0972]$ and $[0.2140,0.4067]$, respectively. In particular, the credible interval for $\nu$ lies well below one, indicating overdispersion relative to the Poisson submodel. Figure~\ref{fig:cmp_insurance} compares the observed claim frequencies with those implied by the DFD-BETEL posterior. The fitted CMP distribution closely reproduces the observed distribution, particularly for zero, one, and two claims. Overall, this analysis illustrates the use of DFD-BETEL for Bayesian inference and uncertainty quantification in a real-world count model with an intractable normalizing constant, without requiring learning-rate calibration.

\begin{figure}[t]
\centering
\includegraphics[width=0.78\textwidth]
{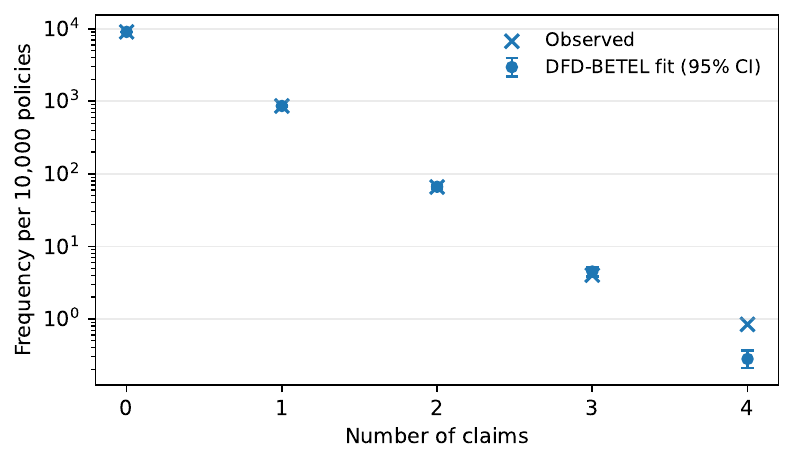}
\caption{Observed and DFD-BETEL fitted frequencies for the automobile claim data. Crosses denote observed frequencies, while filled circles and intervals denote posterior mean fitted frequencies and 95\% credible intervals induced by posterior uncertainty in $(\lambda,\nu)$. The vertical axis is displayed on a logarithmic scale.}
\label{fig:cmp_insurance}
\end{figure}

\section{Conclusion}
\label{sec:discussion}

We propose DFD-BETEL, a semiparametric Bayesian framework for discrete models with intractable normalizing constants. By incorporating estimating equations derived from discrete Fisher divergence into Bayesian ETEL, DFD-BETEL avoids evaluation of the normalizing constant and requires no learning-rate calibration. We establish consistency and asymptotic normality of the sample DFD estimator and a Bernstein--von Mises theorem for the DFD-BETEL posterior, whose limiting covariance matches the sampling covariance of the estimator. Numerical studies show that DFD-BETEL remains competitive under correct specification and provides reliable uncertainty quantification under model misspecification without learning-rate calibration. The automobile-insurance application further illustrates its use for Bayesian inference in real-world count data. Our results also suggest several directions for future work. DFD-based procedures may sacrifice efficiency relative to likelihood-based inference under correct specification, motivating a more systematic study of the efficiency of DFD estimating equations across discrete model classes. Extending DFD-BETEL to high-dimensional models, dependent observations, and alternative discrete score constructions provides promising directions for future research.

\section*{AI Use Statement}
Generative AI tools assisted with developing and debugging
simulation and analysis code, performing computational
checks, and reviewing mathematical assumptions and proof arguments. They also assisted with manuscript editing. The authors reviewed the AI-assisted content and take full responsibility for the final methodology, mathematical claims and proofs, experimental results, interpretations, and manuscript content.

\bibliography{dfd_betel}

\clearpage
\appendix
\counterwithin{theorem}{section}
\counterwithin{lemma}{section}
% \counterwithin{figure}{section}
% \counterwithin{table}{section}
% \counterwithin{algorithm}{section}
\section*{Supplementary Materials}
\addcontentsline{toc}{section}{Supplementary Materials}

\section{Theoretical Results and Proofs}
\label{app:theory}

\subsection{Consistency of the Sample DFD Estimator}
\label{app:consistency}

We first establish consistency of the sample DFD estimator $\widehat{\boldsymbol{\theta}}_n$, defined as a minimizer of the empirical DFD criterion $\widehat J_n(\boldsymbol{\theta})
=n^{-1}\sum_{i=1}^n \varphi(\mathbf{x}_i,\boldsymbol{\theta})$.

\begin{theorem}[Consistency of the sample DFD estimator]
\label{thm:dfd_consistency}
Suppose $\mathbf{x}_1,\ldots,\mathbf{x}_n$ are i.i.d. from $p_{\mathbf{x}}$. Under Assumption~\ref{ass:identification}, the sample DFD estimator satisfies
\[
\widehat{\boldsymbol{\theta}}_n
\overset{p}{\longrightarrow}
\boldsymbol{\theta}^*
\qquad \text{as } n\to\infty.
\]
\end{theorem}

\begin{proof}
Let $P_n f=n^{-1}\sum_{i=1}^n f(\mathbf{x}_i)$ and
$P f=\mathbb{E}_{p_{\mathbf{x}}}[f(\mathbf{x})]$. Then
$\widehat J_n(\boldsymbol{\theta})
=P_n\varphi(\cdot,\boldsymbol{\theta})$ and
$J(\boldsymbol{\theta})
=P\varphi(\cdot,\boldsymbol{\theta})$. We first establish uniform convergence of the empirical DFD criterion to its population counterpart. Define
\[
F(\mathbf{x}) = \sup_{\boldsymbol{\theta}\in\Theta}|\varphi(\mathbf{x},\boldsymbol{\theta})|.
\]
By Assumption~\ref{ass:identification}, $\mathbb{E}_{p_{\mathbf{x}}}[F(\mathbf{x})]<\infty$. Since $\varphi(\mathbf{x},\boldsymbol{\theta})$ is continuous in $\boldsymbol{\theta}$ for each $\mathbf{x}$ and $\Theta$ is compact, $\varphi(\mathbf{x},\cdot)$ is uniformly continuous on $\Theta$. For $\delta>0$, define
\[
w_\delta(\mathbf{x}) = \sup_{\substack{\boldsymbol{\theta},\boldsymbol{\eta}\in\Theta\\
\|\boldsymbol{\theta}-\boldsymbol{\eta}\|\leq\delta}}
|\varphi(\mathbf{x},\boldsymbol{\theta})
-\varphi(\mathbf{x},\boldsymbol{\eta})|.
\]
For each $\mathbf{x}$, $w_\delta(\mathbf{x})\to0$ as $\delta\downarrow0$, while $w_\delta(\mathbf{x})\leq2F(\mathbf{x})$. Dominated convergence therefore gives
\[
\mathbb{E}_{p_{\mathbf{x}}}[w_\delta(\mathbf{x})]
\longrightarrow0
\qquad (\delta\downarrow0).
\]
For fixed $\delta>0$, compactness of $\Theta$ guarantees a finite $\delta$-net $\{\boldsymbol{\theta}_1,\ldots, \boldsymbol{\theta}_{K_\delta}\}$. For every $\boldsymbol{\theta}\in\Theta$, choose $\boldsymbol{\theta}_k$ satisfying $\|\boldsymbol{\theta}-\boldsymbol{\theta}_k\|\leq\delta$. It follows that
\[
\begin{aligned}
\sup_{\boldsymbol{\theta}\in\Theta}
|(P_n-P)\varphi(\cdot,\boldsymbol{\theta})|
\leq {}&
\max_{1\leq k\leq K_\delta}
|(P_n-P)\varphi(\cdot,\boldsymbol{\theta}_k)| \\
&+P_nw_\delta
+\mathbb{E}_{p_{\mathbf{x}}}[w_\delta(\mathbf{x})].
\end{aligned}
\]
For each fixed $\delta$, the strong law of large numbers implies
\[
\max_{1\leq k\leq K_\delta}
|(P_n-P)\varphi(\cdot,\boldsymbol{\theta}_k)|
\longrightarrow0
\]
almost surely, and $P_nw_\delta\to \mathbb{E}_{p_{\mathbf{x}}}[w_\delta(\mathbf{x})]$ almost surely. Now set $\delta_m=1/m$, $m\geq1$. On the intersection of the corresponding probability-one events,
\[
\limsup_{n\to\infty}
\sup_{\boldsymbol{\theta}\in\Theta}
|(P_n-P)\varphi(\cdot,\boldsymbol{\theta})|
\leq
2\mathbb{E}_{p_{\mathbf{x}}}[w_{\delta_m}(\mathbf{x})]
\]
for every $m$. Letting $m\to\infty$ yields
\[
\sup_{\boldsymbol{\theta}\in\Theta}
|\widehat J_n(\boldsymbol{\theta})
-J(\boldsymbol{\theta})|
\longrightarrow0
\]
almost surely. We next use identification of the population minimizer. Pointwise continuity of $\varphi(\mathbf{x},\boldsymbol{\theta})$ together with the integrable envelope $F(\mathbf{x})$ implies continuity of $J(\boldsymbol{\theta})$ by dominated convergence. For any $\epsilon>0$, define
\[
A_\epsilon = \{\boldsymbol{\theta}\in\Theta: \|\boldsymbol{\theta}-\boldsymbol{\theta}^*\|\geq\epsilon\}.
\]
If $A_\epsilon$ is nonempty, compactness of $\Theta$, continuity of $J$, and uniqueness of $\boldsymbol{\theta}^*$ imply
\[
c_\epsilon =
\inf_{\boldsymbol{\theta}\in A_\epsilon}
\{J(\boldsymbol{\theta})-J(\boldsymbol{\theta}^*)\}
>0.
\]
Since $\widehat{\boldsymbol{\theta}}_n$ minimizes $\widehat J_n$, $\widehat J_n(\widehat{\boldsymbol{\theta}}_n)\leq\widehat J_n(\boldsymbol{\theta}^*)$. Hence,
\[
\begin{aligned}
0
\leq
J(\widehat{\boldsymbol{\theta}}_n)-J(\boldsymbol{\theta}^*)
&\leq
|J(\widehat{\boldsymbol{\theta}}_n)
-\widehat J_n(\widehat{\boldsymbol{\theta}}_n)| \\
&\quad+
|\widehat J_n(\boldsymbol{\theta}^*)
-J(\boldsymbol{\theta}^*)| \\
&\leq
2\sup_{\boldsymbol{\theta}\in\Theta}
|\widehat J_n(\boldsymbol{\theta})
-J(\boldsymbol{\theta})|.
\end{aligned}
\]
The right-hand side converges almost surely to zero. On the event $\widehat{\boldsymbol{\theta}}_n\in A_\epsilon$, however, the left-hand side is at least $c_\epsilon$. Therefore,
\[
\mathbb{P}
\left(
\|\widehat{\boldsymbol{\theta}}_n
-\boldsymbol{\theta}^*\|\geq\epsilon
\right)
\longrightarrow0.
\]
Since $\epsilon>0$ is arbitrary, $\widehat{\boldsymbol{\theta}}_n \overset{p}{\longrightarrow}\boldsymbol{\theta}^*$.
\end{proof}

\subsection{Asymptotic Normality of the Sample DFD Estimator}
\label{app:asymptotic_normality}

\begin{theorem}[Asymptotic normality of the sample DFD estimator]
\label{thm:dfd_asymptotic_normality}
Suppose $\mathbf{x}_1,\ldots,\mathbf{x}_n$ are i.i.d. from $p_{\mathbf{x}}$. Under Assumptions~\ref{ass:identification}--\ref{ass:nonsingularity},
\[
\sqrt{n}
(\widehat{\boldsymbol{\theta}}_n-\boldsymbol{\theta}^*)
\overset{d}{\longrightarrow}
N\left(
\mathbf{0}_p,
\boldsymbol{\Gamma}^{-1}
\boldsymbol{\Delta}
\boldsymbol{\Gamma}^{-\top}
\right),
\]
where $\overset{d}{\longrightarrow}$ denotes convergence in distribution, $\boldsymbol{\Delta} =\mathbb{E}_{p_{\mathbf{x}}}
[\mathbf{g}(\mathbf{x},\boldsymbol{\theta}^*)\mathbf{g}(\mathbf{x},\boldsymbol{\theta}^*)^\top]$, and $\boldsymbol{\Gamma} =\mathbb{E}_{p_{\mathbf{x}}}[\partial\mathbf{g}(\mathbf{x},\boldsymbol{\theta}^*)/
\partial\boldsymbol{\theta}^{\top}]$.
\end{theorem}

\begin{proof}
Let $G(\mathbf{x},\boldsymbol{\theta}) =\partial\mathbf{g}(\mathbf{x},\boldsymbol{\theta})/\partial\boldsymbol{\theta}^{\top}$. By Theorem~\ref{thm:dfd_consistency}, $\widehat{\boldsymbol{\theta}}_n \overset{p}{\longrightarrow}\boldsymbol{\theta}^*$. Since $\boldsymbol{\theta}^*\in\operatorname{int}(\Theta)$ and $\mathcal{N}$ in Assumption~\ref{ass:regularity} is an open neighborhood of $\boldsymbol{\theta}^*$, there exists a closed ball $\mathcal{B}$ centered at $\boldsymbol{\theta}^*$ such that $\mathcal{B}\subset\operatorname{int}(\Theta)\cap\mathcal{N}$. Consistency therefore implies $\mathbb{P}(\widehat{\boldsymbol{\theta}}_n\in\mathcal{B})\to1$. All subsequent first-order and Taylor expansions are understood on this event, whose probability tends to one. The derivative envelope in Assumption~\ref{ass:regularity} permits differentiation under the expectation. Since $\boldsymbol{\theta}^*$ is an interior minimizer of $J(\boldsymbol{\theta})$, its population first-order condition is
\[
\mathbb{E}_{p_{\mathbf{x}}}[\mathbf{g}(\mathbf{x},\boldsymbol{\theta}^*)]
=
\mathbf{0}_p.
\]
Likewise, because $\widehat{\boldsymbol{\theta}}_n\in\operatorname{int}(\Theta)$ on the event above and minimizes the empirical DFD criterion,
\[
P_n\mathbf{g}
(\cdot,\widehat{\boldsymbol{\theta}}_n)
=
\mathbf{0}_p,
\]
where $P_n f=n^{-1}\sum_{i=1}^n f(\mathbf{x}_i)$. For $t\in[0,1]$, define $\boldsymbol{\theta}_{n,t} =\boldsymbol{\theta}^* +t(\widehat{\boldsymbol{\theta}}_n-\boldsymbol{\theta}^*)$. Since $\mathcal{B}$ is convex, $\boldsymbol{\theta}_{n,t}\in\mathcal{B}$ whenever $\widehat{\boldsymbol{\theta}}_n\in\mathcal{B}$. The integral form of Taylor's theorem gives
\[
\mathbf{0}_p =
P_n\mathbf{g}(\cdot,\boldsymbol{\theta}^*)
+
\widehat{\boldsymbol{\Gamma}}_n
(\widehat{\boldsymbol{\theta}}_n-\boldsymbol{\theta}^*),
\]
where
\[
\widehat{\boldsymbol{\Gamma}}_n
=
\int_0^1
P_nG(\cdot,\boldsymbol{\theta}_{n,t})\,dt.
\]
We next show that $\widehat{\boldsymbol{\Gamma}}_n \overset{p}{\longrightarrow}\boldsymbol{\Gamma}$. On the event $\widehat{\boldsymbol{\theta}}_n\in\mathcal{B}$,
\[
\begin{aligned}
\|
\widehat{\boldsymbol{\Gamma}}_n-\boldsymbol{\Gamma}
\|_F
\leq {}&
\sup_{\boldsymbol{\theta}\in\mathcal{B}}
\|
(P_n-P)G(\cdot,\boldsymbol{\theta})
\|_F \\
&+
\sup_{t\in[0,1]}
\|
PG(\cdot,\boldsymbol{\theta}_{n,t})
-PG(\cdot,\boldsymbol{\theta}^*)
\|_F,
\end{aligned}
\]
where $Pf=\mathbb{E}_{p_{\mathbf{x}}}[f(\mathbf{x})]$. The first term converges to zero in probability by the uniform law of large numbers, using the continuity and integrable derivative envelope in Assumption~\ref{ass:regularity}. For the second term,
\[
\sup_{t\in[0,1]}
\|
\boldsymbol{\theta}_{n,t}-\boldsymbol{\theta}^*
\|
\leq
\|
\widehat{\boldsymbol{\theta}}_n-\boldsymbol{\theta}^*
\|
\overset{p}{\longrightarrow}0.
\]
Continuity of $G(\mathbf{x},\boldsymbol{\theta})$ and the same integrable envelope imply, by dominated convergence, that $PG(\cdot,\boldsymbol{\theta})$ is continuous at $\boldsymbol{\theta}^*$. Hence the second term also converges to zero in probability, and therefore
\[
\widehat{\boldsymbol{\Gamma}}_n
\overset{p}{\longrightarrow}
PG(\cdot,\boldsymbol{\theta}^*)
=
\boldsymbol{\Gamma}.
\]
By Assumption~\ref{ass:nonsingularity}, $\boldsymbol{\Gamma}$ has full rank $p$ and is therefore nonsingular. It follows that $\widehat{\boldsymbol{\Gamma}}_n$ is nonsingular with probability tending to one and $\widehat{\boldsymbol{\Gamma}}_n^{-1}\overset{p}{\longrightarrow}\boldsymbol{\Gamma}^{-1}$. Rearranging the Taylor expansion yields
\[
\sqrt{n}
(\widehat{\boldsymbol{\theta}}_n-\boldsymbol{\theta}^*)
=
-\widehat{\boldsymbol{\Gamma}}_n^{-1}
\frac{1}{\sqrt{n}}
\sum_{i=1}^n
\mathbf{g}(\mathbf{x}_i,\boldsymbol{\theta}^*).
\]
Equivalently,
\[
\sqrt{n}
(\widehat{\boldsymbol{\theta}}_n-\boldsymbol{\theta}^*)
=
-\boldsymbol{\Gamma}^{-1}
\frac{1}{\sqrt{n}}
\sum_{i=1}^n
\mathbf{g}(\mathbf{x}_i,\boldsymbol{\theta}^*)
+o_p(1).
\]
By Assumption~\ref{ass:regularity}, $\mathbf{g}(\mathbf{x},\boldsymbol{\theta}^*)$ has finite second moments, and its expectation is $\mathbf{0}_p$. The multivariate central limit theorem therefore gives
\[
\frac{1}{\sqrt{n}}
\sum_{i=1}^n
\mathbf{g}(\mathbf{x}_i,\boldsymbol{\theta}^*)
\overset{d}{\longrightarrow}
N(\mathbf{0}_p,\boldsymbol{\Delta}).
\]
Applying Slutsky's theorem to the asymptotic linear representation gives
\[
\sqrt{n}
(\widehat{\boldsymbol{\theta}}_n-\boldsymbol{\theta}^*)
\overset{d}{\longrightarrow}
N\left(
\mathbf{0}_p,
\boldsymbol{\Gamma}^{-1}
\boldsymbol{\Delta}
\boldsymbol{\Gamma}^{-\top}
\right),
\]
which proves the result.
\end{proof}

\subsection{Proof of the Bernstein--von Mises Theorem}

\label{app:bvm}

The proof follows the Bernstein--von Mises argument for Bayesian ETEL in \citet{chib2018bayesian}, adapted to the DFD estimating function $\mathbf{g}(\mathbf{x},\boldsymbol{\theta})
=\nabla_{\boldsymbol{\theta}}\varphi(\mathbf{x},\boldsymbol{\theta})$. Although the working probability model may be misspecified, $\boldsymbol{\theta}^*$ is an interior minimizer of the population DFD criterion and therefore satisfies
\[
\mathbb{E}_{p_{\mathbf{x}}}
[\mathbf{g}(\mathbf{x},\boldsymbol{\theta}^*)]
=
\mathbf{0}_p.
\]
Thus, the moment restriction defining the ETEL posterior is correctly specified at $\boldsymbol{\theta}^*$. Define
\[
\boldsymbol{V}
=
\boldsymbol{\Gamma}^{\top}
\boldsymbol{\Delta}^{-1}
\boldsymbol{\Gamma},
\qquad
\mathbf{h}
=
\sqrt{n}
(\boldsymbol{\theta}-\widehat{\boldsymbol{\theta}}_n),
\]
and
\[
R_n(\mathbf{h})
=
\frac{
L(\widehat{\boldsymbol{\theta}}_n+\mathbf{h}/\sqrt{n})
}{
L(\widehat{\boldsymbol{\theta}}_n)
}.
\]
We set $R_n(\mathbf{h})=0$ whenever $\widehat{\boldsymbol{\theta}}_n+\mathbf{h}/\sqrt{n}\notin\Theta$. By Theorem~\ref{thm:dfd_consistency}, $\widehat{\boldsymbol{\theta}}_n \overset{p}{\to}\boldsymbol{\theta}^*$. Since $\boldsymbol{\theta}^*\in\operatorname{int}(\Theta)$, with probability tending to one $\widehat{\boldsymbol{\theta}}_n$ is an interior minimizer of the empirical DFD criterion and hence
\[
P_n\mathbf{g}
(\cdot,\widehat{\boldsymbol{\theta}}_n)
=
\mathbf{0}_p.
\]
Local uniqueness of the ETEL multiplier then gives $\widehat{\boldsymbol{\lambda}}(\widehat{\boldsymbol{\theta}}_n)=\mathbf{0}_p$. Consequently, the ETEL weights at $\widehat{\boldsymbol{\theta}}_n$ are $1/n$ and $L(\widehat{\boldsymbol{\theta}}_n)=n^{-n}$.

\begin{lemma}
\label{lem:dfd_betel_local}

Under Assumptions~\ref{ass:identification}--\ref{ass:posterior}, for every fixed $K<\infty$,
\[
\sup_{\|\mathbf{h}\|\leq K}
\left|
\log R_n(\mathbf{h})
+
\frac12
\mathbf{h}^{\top}
\boldsymbol{V}\mathbf{h}
\right|
\overset{p}{\longrightarrow}0.
\]
\end{lemma}

\begin{proof}
Fix $K<\infty$. Uniformly over $\|\mathbf{h}\|\leq K$,
\[
\widehat{\boldsymbol{\theta}}_n
+\frac{\mathbf{h}}{\sqrt n}
=
\boldsymbol{\theta}^*+o_p(1).
\]
Hence, with probability tending to one, these local parameter values lie in a neighborhood of $\boldsymbol{\theta}^*$ on which the ETEL multiplier is well-defined and unique. At $\widehat{\boldsymbol{\theta}}_n$, the empirical moment condition is exactly satisfied, so
$\widehat{\boldsymbol{\lambda}}(\widehat{\boldsymbol{\theta}}_n)=\mathbf{0}_p$.
Let
\[
G(\mathbf{x},\boldsymbol{\theta})
=
\frac{\partial
\mathbf{g}(\mathbf{x},\boldsymbol{\theta})}
{\partial\boldsymbol{\theta}^{\top}}.
\]
Differentiating the ETEL multiplier equation at $\widehat{\boldsymbol{\theta}}_n$ gives
\[
\frac{
\partial\widehat{\boldsymbol{\lambda}}
(\widehat{\boldsymbol{\theta}}_n)
}{
\partial\boldsymbol{\theta}^{\top}
}
=
-
\widehat{\boldsymbol{\Delta}}_n^{-1}
\widehat{\boldsymbol{\Gamma}}_n,
\]
where
\[
\widehat{\boldsymbol{\Delta}}_n
=
P_n\left[
\mathbf{g}(\cdot,\widehat{\boldsymbol{\theta}}_n)
\mathbf{g}(\cdot,\widehat{\boldsymbol{\theta}}_n)^\top
\right],
\qquad
\widehat{\boldsymbol{\Gamma}}_n
=
P_nG(\cdot,\widehat{\boldsymbol{\theta}}_n).
\]
By Theorem~\ref{thm:dfd_consistency}, the uniform laws of large numbers implied by Assumption~\ref{ass:regularity}, and Assumption~\ref{ass:nonsingularity},
\[
\widehat{\boldsymbol{\Delta}}_n
\overset{p}{\longrightarrow}
\boldsymbol{\Delta},
\qquad
\widehat{\boldsymbol{\Gamma}}_n
\overset{p}{\longrightarrow}
\boldsymbol{\Gamma}.
\]
Therefore,
\[
\frac{
\partial\widehat{\boldsymbol{\lambda}}
(\widehat{\boldsymbol{\theta}}_n)
}{
\partial\boldsymbol{\theta}^{\top}
}
\overset{p}{\longrightarrow}
-\boldsymbol{\Delta}^{-1}\boldsymbol{\Gamma}.
\]
The local ETEL expansion of \citet{chib2018bayesian}, applied to the DFD estimating function under Assumptions~\ref{ass:regularity}--\ref{ass:posterior}, then yields, uniformly over bounded $\mathbf{h}$,
\[
\ell_n\left(
\widehat{\boldsymbol{\theta}}_n
+\frac{\mathbf{h}}{\sqrt n}
\right)
-
\ell_n(\widehat{\boldsymbol{\theta}}_n)
=
-\frac12
\mathbf{h}^{\top}
\boldsymbol{\Gamma}^{\top}
\boldsymbol{\Delta}^{-1}
\boldsymbol{\Gamma}
\mathbf{h}
+
o_p(1).
\]
Since $\boldsymbol{V} =\boldsymbol{\Gamma}^{\top}\boldsymbol{\Delta}^{-1}\boldsymbol{\Gamma}$, the result follows.
\end{proof}

\begin{lemma}
\label{lem:dfd_betel_integrated}

Under Assumptions~\ref{ass:identification}--\ref{ass:posterior},
\[
\int_{\mathbb{R}^p}
\left|
\pi\left(
\widehat{\boldsymbol{\theta}}_n
+\frac{\mathbf{h}}{\sqrt n}
\right)
R_n(\mathbf{h})
-
\pi(\boldsymbol{\theta}^*)
\exp\left(
-\frac12
\mathbf{h}^{\top}
\boldsymbol{V}\mathbf{h}
\right)
\right|
d\mathbf{h}
\overset{p}{\longrightarrow}0,
\]
where the first term in the integrand is defined to be zero whenever $\widehat{\boldsymbol{\theta}}_n+\mathbf{h}/\sqrt n\notin\Theta$.
\end{lemma}

\begin{proof}
The argument follows the localization argument for Bayesian ETEL in \citet{chib2018bayesian}. Fix $\delta>0$ sufficiently small that the $\delta$-neighborhood of $\boldsymbol{\theta}^*$ lies in the region on which the local ETEL regularity conditions hold. For fixed $K<\infty$, partition $\mathbb{R}^p$ into
\[
A_{1n}
=
\{\|\mathbf{h}\|\leq K\},
\]
\[
A_{2n}
=
\left\{
\|\mathbf{h}\|>K,\,
\left\|
\widehat{\boldsymbol{\theta}}_n
+\frac{\mathbf{h}}{\sqrt n}
-\boldsymbol{\theta}^*
\right\|
\leq\delta
\right\},
\]
and
\[
A_{3n}
=
\left\{
\left\|
\widehat{\boldsymbol{\theta}}_n
+\frac{\mathbf{h}}{\sqrt n}
-\boldsymbol{\theta}^*
\right\|
>\delta
\right\}.
\]
On $A_{1n}$, Lemma~\ref{lem:dfd_betel_local} gives
\[
\sup_{\mathbf{h}\in A_{1n}}
\left|
\log R_n(\mathbf{h})
+
\frac12
\mathbf{h}^{\top}\boldsymbol{V}\mathbf{h}
\right|
\overset{p}{\longrightarrow}0.
\]
Moreover,
\[
\sup_{\mathbf{h}\in A_{1n}}
\left\|
\widehat{\boldsymbol{\theta}}_n
+\frac{\mathbf{h}}{\sqrt n}
-\boldsymbol{\theta}^*
\right\|
\overset{p}{\longrightarrow}0.
\]
Continuity of the prior density at $\boldsymbol{\theta}^*$ therefore implies
\[
\sup_{\mathbf{h}\in A_{1n}}
\left|
\pi\left(
\widehat{\boldsymbol{\theta}}_n
+\frac{\mathbf{h}}{\sqrt n}
\right)
-
\pi(\boldsymbol{\theta}^*)
\right|
\overset{p}{\longrightarrow}0.
\]
Since $A_{1n}$ has finite Lebesgue measure for fixed $K$, the contribution from $A_{1n}$ converges to zero in probability. For $A_{2n}$, the intermediate-region localization bound in the Bayesian ETEL argument of \citet{chib2018bayesian}, applied to the DFD moment function under Assumptions~\ref{ass:regularity}--\ref{ass:posterior}, gives, for every $\varepsilon>0$,
\[
\lim_{K\to\infty}
\limsup_{n\to\infty}
P\left\{
\int_{A_{2n}}
\pi\left(
\widehat{\boldsymbol{\theta}}_n
+\frac{\mathbf{h}}{\sqrt n}
\right)
R_n(\mathbf{h})
\,d\mathbf{h}
>\varepsilon
\right\}
=
0.
\]
Since $\boldsymbol{\Delta}$ is positive definite and $\boldsymbol{\Gamma}$ is nonsingular, $\boldsymbol{V}$ is positive definite. Hence
\[
\int_{\|\mathbf{h}\|>K}
\exp\left(
-\frac12
\mathbf{h}^{\top}\boldsymbol{V}\mathbf{h}
\right)
d\mathbf{h}
\longrightarrow0
\]
as $K\to\infty$. Thus the contribution from $A_{2n}$ vanishes after first letting $n\to\infty$ and then $K\to\infty$. Finally, consider $A_{3n}$. By the separation condition in Assumption~\ref{ass:posterior}, there exists $\epsilon>0$ such that, with probability tending to one,
\[
\sup_{\|\boldsymbol{\theta}-\boldsymbol{\theta}^*\|>\delta}
\{
\ell_n(\boldsymbol{\theta})
-
\ell_n(\boldsymbol{\theta}^*)
\}
\leq
-n\epsilon.
\]
At $\widehat{\boldsymbol{\theta}}_n$, the ETEL weights are $1/n$ and $L(\widehat{\boldsymbol{\theta}}_n)=n^{-n}$. By the arithmetic--geometric mean inequality, $L(\boldsymbol{\theta})\leq n^{-n}$ for every feasible $\boldsymbol{\theta}$, so
\[
\ell_n(\boldsymbol{\theta}^*)
\leq
\ell_n(\widehat{\boldsymbol{\theta}}_n).
\]
Consequently, uniformly over $A_{3n}$,
\[
R_n(\mathbf{h})
\leq
e^{-n\epsilon}
\]
with probability tending to one. Using the change of variables $\boldsymbol{\theta} =\widehat{\boldsymbol{\theta}}_n+\mathbf{h}/\sqrt n$,
\[
\begin{aligned}
&\int_{A_{3n}}
\pi\left(
\widehat{\boldsymbol{\theta}}_n
+\frac{\mathbf{h}}{\sqrt n}
\right)
R_n(\mathbf{h})
\,d\mathbf{h}
\\
&\qquad\leq
n^{p/2}e^{-n\epsilon}
\int_{\Theta}\pi(\boldsymbol{\theta})\,d\boldsymbol{\theta}
=
o_p(1).
\end{aligned}
\]
By consistency, with probability tending to one, $\|\widehat{\boldsymbol{\theta}}_n-\boldsymbol{\theta}^*\|
<\delta/2$. On this event, $\mathbf{h}\in A_{3n}$ implies $\|\mathbf{h}\|>\delta\sqrt n/2$, and therefore
\[
\int_{A_{3n}}
\exp\left(
-\frac12
\mathbf{h}^{\top}\boldsymbol{V}\mathbf{h}
\right)
d\mathbf{h}
\longrightarrow0.
\]
Combining the three regions, first letting $n\to\infty$ for fixed $K$ and then $K\to\infty$, proves the result.
\end{proof}

\begin{proof}[Proof of Theorem~\ref{thm:bvm}]

Define
\[
k_n(\mathbf{h})
=
\pi\left(
\widehat{\boldsymbol{\theta}}_n
+\frac{\mathbf{h}}{\sqrt n}
\right)
R_n(\mathbf{h}),
\qquad
C_n
=
\int_{\mathbb{R}^p}
k_n(\mathbf{h})\,d\mathbf{h},
\]
with $k_n(\mathbf{h})=0$ outside the transformed parameter space. The posterior density of $\mathbf{h} =\sqrt n(\boldsymbol{\theta} -\widehat{\boldsymbol{\theta}}_n)$ is
\[
q_n(\mathbf{h}\mid\mathbf{x}_{1:n})
=
C_n^{-1}k_n(\mathbf{h}).
\]
By Lemma~\ref{lem:dfd_betel_integrated},
\[
C_n
\overset{p}{\longrightarrow}
\pi(\boldsymbol{\theta}^*)
(2\pi)^{p/2}
|\boldsymbol{V}|^{-1/2}.
\]
The limiting constant is strictly positive because $\pi(\boldsymbol{\theta}^*)>0$ and $\boldsymbol{V}$ is positive definite. Let
\[
\phi_{\boldsymbol{V}}(\mathbf{h})
=
(2\pi)^{-p/2}
|\boldsymbol{V}|^{1/2}
\exp\left(
-\frac12
\mathbf{h}^{\top}\boldsymbol{V}\mathbf{h}
\right).
\]
Lemma~\ref{lem:dfd_betel_integrated}, together with the convergence of $C_n$, gives
\[
\int_{\mathbb{R}^p}
\left|
q_n(\mathbf{h}\mid\mathbf{x}_{1:n})
-
\phi_{\boldsymbol{V}}(\mathbf{h})
\right|
d\mathbf{h}
\overset{p}{\longrightarrow}0.
\]
Thus the total variation distance between the posterior distribution of $\sqrt n(\boldsymbol{\theta}-\widehat{\boldsymbol{\theta}}_n)$ and $N(\mathbf{0}_p,\boldsymbol{V}^{-1})$ converges to zero in probability. Since
\[
\boldsymbol{V}^{-1}
=
\left(
\boldsymbol{\Gamma}^{\top}
\boldsymbol{\Delta}^{-1}
\boldsymbol{\Gamma}
\right)^{-1},
\]
the conclusion of Theorem~\ref{thm:bvm} follows.
\end{proof}

\section{Posterior Computation}
\label{app:algorithm}

Algorithm~\ref{alg:mh_dfd_betel} summarizes the Metropolis--Hastings sampler used for the DFD-BETEL posterior defined in equation \eqref{eq:dfd_betel_posterior}. Given a proposal density
$k(\widetilde{\boldsymbol{\theta}}\mid\boldsymbol{\theta})$, the
acceptance probability is
\[
\alpha(\boldsymbol{\theta},\widetilde{\boldsymbol{\theta}})
=
\min\left\{
1,
\frac{
L(\widetilde{\boldsymbol{\theta}})
\pi(\widetilde{\boldsymbol{\theta}})
k(\boldsymbol{\theta}\mid\widetilde{\boldsymbol{\theta}})
}{
L(\boldsymbol{\theta})
\pi(\boldsymbol{\theta})
k(\widetilde{\boldsymbol{\theta}}\mid\boldsymbol{\theta})
}
\right\}.
\]

\begin{algorithm}[htbp]
\caption{Metropolis--Hastings sampler for DFD-BETEL}
\label{alg:mh_dfd_betel}
\begin{algorithmic}[1]
\STATE \textbf{Input:} observations $\mathbf{x}_{1:n}$, iterations $S$,
proposal density $k$, and feasible initial state
$\boldsymbol{\theta}^0$.
\FOR{$s=0,\ldots,S-1$}
    \STATE Propose
    $\widetilde{\boldsymbol{\theta}}
    \sim k(\cdot\mid\boldsymbol{\theta}^s)$.
    \IF{$\pi(\widetilde{\boldsymbol{\theta}})=0$ or
    $\mathbf{0}\notin\operatorname{int}
    \mathcal{C}(\widetilde{\boldsymbol{\theta}})$}
        \STATE Set
        $\boldsymbol{\theta}^{s+1}=\boldsymbol{\theta}^s$.
    \ELSE
        \STATE Compute $L(\widetilde{\boldsymbol{\theta}})$.
        \STATE Generate $u \sim \text{Unif}(0,1)$.
        \STATE Set
        $\boldsymbol{\theta}^{s+1}
        =\widetilde{\boldsymbol{\theta}}$
        if
        $u < \alpha(\boldsymbol{\theta}^s,
        \widetilde{\boldsymbol{\theta}})$;
        otherwise set
        $\boldsymbol{\theta}^{s+1}=\boldsymbol{\theta}^s$.
    \ENDIF
\ENDFOR
\STATE \textbf{Output:} posterior draws
$\{\boldsymbol{\theta}^s\}_{s=1}^S$.
\end{algorithmic}
\end{algorithm}

\section{Computational Details}
\label{app:computation}

This section summarizes the main computational settings used in the simulation studies. Model-specific data-generating mechanisms, priors, and pseudo-true targets are provided with the corresponding experiments.

\paragraph{Posterior sampling.}
The correctly specified CMP, Ising, and Poisson graphical-model experiments use four chains of 25,000 iterations with 5,000 warm-up iterations. The Potts experiments use four chains of 15,000 iterations with 3,000 warm-up iterations.

\paragraph{Auxiliary simulation and likelihood approximation.}
For the Poisson graphical model, DMH generates auxiliary datasets using 30 Gibbs sweeps initialized at the observed data. For the Potts model, DMH uses 2,000 random-site heat-bath updates per auxiliary lattice, initialized at the observed lattice. For the CMP experiment, the Exchange sampler generates independent auxiliary observations from the CMP distribution truncated to $\{0,\ldots,200\}$. For standard Bayes, the CMP normalizing series is truncated at $K=\max\{99,\max_i x_i\}$.

\paragraph{Computing infrastructure.}
All simulation experiments were run on the computing cluster using Slurm, with R~4.2.2, Rcpp/C++, and Python~3 for high-precision numerical calculations.

\section{Additional Results under Correct Specification}
\label{app:correct_specification}

This section provides additional numerical results under correct model specification. We first consider the Conway--Maxwell--Poisson (CMP) and Ising models, which complement the Poisson graphical model studied in Section~\ref{sec:pg2}. We then provide additional implementation details for the Poisson graphical model.

\subsection{Conway--Maxwell--Poisson Model}
\label{app:cmp_correct}

As a simulation counterpart to the real-data analysis in Section~\ref{sec:real_data}, we consider the CMP model \citep{shmueli2005useful} defined in \eqref{eq:cmp_model}. We examine two correctly specified settings: $(\lambda_0,\nu_0)=(4,1.25)$ (underdispersion) and $(\lambda_0,\nu_0)=(4,0.75)$ (overdispersion), referred to as Cases I and II, respectively. For each setting, we generate $n=2000$ independent observations and assign independent $\chi_3^2$ priors to $\lambda$ and $\nu$. For the likelihood-based comparison, we use exchange sampling \citep{murray2006mcmc}, with auxiliary observations generated independently from the CMP distribution truncated to $\{0,\ldots,200\}$. Since the normalizing series can also be accurately approximated in this example, standard Bayesian inference is included as an additional likelihood-based benchmark. Table~\ref{tab:cmp_correct} summarizes the estimation results. DFD-BETEL and DFD-Bayes exhibit similar estimation accuracy in both settings, particularly under overdispersion. Exchange and standard Bayes generally attain smaller MAD and MSE, with the efficiency difference more pronounced in Case I. These results are consistent with the Poisson graphical model experiment in Section~\ref{sec:pg2}, where likelihood-based methods are more efficient under correct specification.

\begin{table*}[htbp]
\caption{Correctly specified CMP model: bias, MAD, and MSE of posterior
mean estimates over 500 replications.}
\label{tab:cmp_correct}
\centering
\small
\setlength{\tabcolsep}{5pt}
\begin{tabular}{llrrrrrr}
\toprule
& & \multicolumn{3}{c}{$\lambda$}
  & \multicolumn{3}{c}{$\nu$} \\
\cmidrule(lr){3-5}\cmidrule(lr){6-8}
Case & Method
& Bias & MAD & MSE
& Bias & MAD & MSE \\
\midrule
I & DFD-BETEL
& 0.1322 & 0.3061 & 0.1632
& 0.0172 & 0.0537 & 0.0047 \\
  & DFD-Bayes
& 0.0999 & 0.2919 & 0.1472
& 0.0108 & 0.0518 & 0.0044 \\
  & Exchange
& 0.0371 & 0.1879 & 0.0560
& 0.0043 & 0.0360 & 0.0020 \\
  & Standard Bayes
& 0.0368 & 0.1880 & 0.0560
& 0.0043 & 0.0361 & 0.0020 \\
\midrule
II & DFD-BETEL
& 0.0372 & 0.1762 & 0.0497
& 0.0033 & 0.0218 & 0.0008 \\
   & DFD-Bayes
& 0.0243 & 0.1736 & 0.0483
& 0.0025 & 0.0217 & 0.0007 \\
   & Exchange
& 0.0234 & 0.1690 & 0.0443
& 0.0018 & 0.0211 & 0.0007 \\
   & Standard Bayes
& 0.0234 & 0.1692 & 0.0444
& 0.0018 & 0.0211 & 0.0007 \\
\bottomrule
\end{tabular}
\end{table*}

\subsection{Ising Model}
\label{app:ising_correct}

We next consider the Ising model \citep{ising1925beitrag}, a classical model for interacting binary variables. For an $M\times N$ lattice, let $\mathbf{x}=(x_1,\ldots,x_{MN})$, where $x_i\in\{-1,+1\}$. The probability mass function is
\[
p(\mathbf{x};\alpha,\beta)
=
\frac{1}{Z(\alpha,\beta)}
\exp\left\{
\alpha\sum_{i=1}^{MN}x_i
+
\beta\sum_{(i,j)\in E}x_i x_j
\right\},
\]
where $E$ denotes the set of nearest-neighbor edges, $\alpha$ represents the external field, and $\beta$ controls the interaction between neighboring spins. The normalizing constant involves a sum over all $2^{MN}$ possible configurations and becomes computationally intractable as the lattice size increases.

We use a $10\times10$ lattice with $(\alpha_0,\beta_0)=(0.2,0.3)$. For each replication, we generate $n=50$ independent lattice observations using perfect sampling \citep{propp1996exact} and assign independent $\operatorname{Unif}(0,1)$ priors to $\alpha$ and $\beta$. Table~\ref{tab:ising_correct} summarizes the results. DFD-BETEL has the smallest absolute bias for both parameters, while DFD-Bayes achieves slightly smaller MAD and MSE than DFD-BETEL. DMH attains the smallest MAD and MSE, again illustrating the efficiency advantage of likelihood-based inference under correct specification.

\begin{table}[htbp]
\caption{Correctly specified Ising model: bias, MAD, and MSE of
posterior mean estimates over 500 replications. Bias and MAD are
multiplied by $10^2$, and MSE by $10^3$.}
\label{tab:ising_correct}
\centering
\small
\setlength{\tabcolsep}{4pt}
\begin{tabular}{lrrrrrr}
\toprule
& \multicolumn{2}{c}{Bias}
& \multicolumn{2}{c}{MAD}
& \multicolumn{2}{c}{MSE} \\
\cmidrule(lr){2-3}
\cmidrule(lr){4-5}
\cmidrule(lr){6-7}
Method
& $\alpha$ & $\beta$
& $\alpha$ & $\beta$
& $\alpha$ & $\beta$ \\
\midrule
DFD-BETEL
& -0.03 & 0.11
& 3.90 & 1.76
& 2.34 & 0.48 \\
DFD-Bayes
& 2.65 & -0.90
& 3.43 & 1.56
& 1.78 & 0.38 \\
DMH
& 0.37 & -0.17
& 1.77 & 1.04
& 0.51 & 0.17 \\
\bottomrule
\end{tabular}
\end{table}

\subsection{Additional Details for the Poisson Graphical Model}
\label{app:pg2_correct_details}

For the correctly specified experiment in Section~\ref{sec:pg2}, data are generated from \eqref{eq:pg2_working} with $\boldsymbol{\theta}_0=(\log 8,\log 10,0.05)$. The joint distribution is numerically normalized on $\{0,\ldots,120\}^2$, with the omitted probability verified to be below $10^{-12}$. For the Poisson graphical model, the ratios entering the DFD criterion are
\[
r_{j-}(\mathbf{x},\boldsymbol{\theta})
=
x_j\exp\{-a_j+\gamma_{12}x_{3-j}\},
\qquad
r_{j+}(\mathbf{x},\boldsymbol{\theta})
=
(x_j+1)\exp\{-a_j+\gamma_{12}x_{3-j}\},
\]
for $j=1,2$, where $r_{j-}(\mathbf{x},\boldsymbol{\theta})=0$ when $x_j=0$. The resulting DFD loss is
\[
\varphi(\mathbf{x},\boldsymbol{\theta})
=
\sum_{j=1}^2
\left\{
x_j^2\exp\{-2a_j+2\gamma_{12}x_{3-j}\}
-
2(x_j+1)\exp\{-a_j+\gamma_{12}x_{3-j}\}
\right\}.
\]
DFD-BETEL uses $\nabla_{\boldsymbol{\theta}}
\varphi(\mathbf{x},\boldsymbol{\theta})$ as its estimating function.

\section{Misspecified Poisson Graphical Model}
\label{app:pg2_misspecified}

We further examine uncertainty quantification under misspecification using the two-node Poisson graphical model from Section~\ref{sec:pg2}. Data are generated from
\begin{equation}
\label{eq:pg2_model}
f_\nu(x_1,x_2)
\propto
\exp\{a_{01}x_1+a_{02}x_2-\gamma_0x_1x_2\}
(x_1!x_2!)^{-\nu},
\end{equation}
where $(a_{01},a_{02},\gamma_0)=(\log 8,\log 10,0.05)$. We consider $\nu=0.8$ and $\nu=1.25$ and fit the working model
\eqref{eq:pg2_working}, which fixes $\nu=1$, thereby inducing misspecification through the factorial exponent. For each setting, we consider $n=200,500,$ and $1000$. DFD-BETEL and DFD-Bayes coverage is evaluated with respect to the DFD pseudo-true parameter, while DMH coverage is evaluated with respect to the KL pseudo-true parameter.

\begin{table}[htbp]
\caption{Misspecified Poisson graphical model: empirical coverage
probabilities of nominal 95\% credible intervals over 500 replications.
Each method is evaluated with respect to its corresponding pseudo-true
parameter.}
\label{tab:pg2_misspecified}
\centering
\small
\setlength{\tabcolsep}{4pt}
\begin{tabular}{rrlrrr}
\toprule
$\nu$ & $n$ & Method
& $a_1$ & $a_2$ & $\gamma_{12}$ \\
\midrule
0.8 & 200  & DFD-BETEL & 0.908 & 0.918 & 0.936 \\
    &      & DFD-Bayes & 0.982 & 0.996 & 0.970 \\
    &      & DMH       & 0.926 & 0.922 & 0.930 \\
\addlinespace
    & 500  & DFD-BETEL & 0.928 & 0.940 & 0.928 \\
    &      & DFD-Bayes & 0.990 & 0.998 & 0.970 \\
    &      & DMH       & 0.924 & 0.918 & 0.908 \\
\addlinespace
    & 1000 & DFD-BETEL & 0.940 & 0.944 & 0.934 \\
    &      & DFD-Bayes & 0.990 & 0.996 & 0.956 \\
    &      & DMH       & 0.920 & 0.912 & 0.936 \\
\midrule
1.25 & 200  & DFD-BETEL & 0.904 & 0.944 & 0.918 \\
     &      & DFD-Bayes & 0.932 & 0.952 & 0.910 \\
     &      & DMH       & 0.978 & 0.984 & 0.990 \\
\addlinespace
     & 500  & DFD-BETEL & 0.946 & 0.932 & 0.930 \\
     &      & DFD-Bayes & 0.940 & 0.954 & 0.888 \\
     &      & DMH       & 0.982 & 0.982 & 0.984 \\
\addlinespace
     & 1000 & DFD-BETEL & 0.936 & 0.946 & 0.942 \\
     &      & DFD-Bayes & 0.944 & 0.962 & 0.920 \\
     &      & DMH       & 0.974 & 0.976 & 0.976 \\
\bottomrule
\end{tabular}
\end{table}

Table~\ref{tab:pg2_misspecified} shows that DFD-BETEL coverage
generally approaches the nominal 95\% level as the sample size
increases, reaching $0.934$--$0.946$ across all coordinates at
$n=1000$. Its finite-sample coverage can be lower for smaller
samples. DFD-Bayes exhibits coordinate- and setting-dependent
behavior, including conservative coverage for $\nu=0.8$ and
undercoverage of $\gamma_{12}$ for $\nu=1.25$. DMH also exhibits
setting-dependent miscalibration, tending to undercover for
$\nu=0.8$ and overcover for $\nu=1.25$. These results further
illustrate the calibration of DFD-BETEL under model misspecification while highlighting finite-sample deviations from nominal coverage.

\section{Additional Learning-Rate Sensitivity Results}
\label{app:sensitivity}

We further examine the sensitivity of DFD-Bayes to the learning rate using the misspecified Poisson graphical model in Appendix~\ref{app:pg2_misspecified}. As in Section~\ref{sec:learning_rate}, we consider learning rates $c\widehat{\eta}$ for $c\in\{0.25,0.5,1,2,4\}$, where $c=1$ corresponds to the calibrated baseline. Figure~\ref{fig:pg2_sensitivity} reports coverage at $n=1000$ for both misspecification settings.

\begin{figure}[htbp]
    \centering
    \includegraphics[width=0.9\textwidth]
        {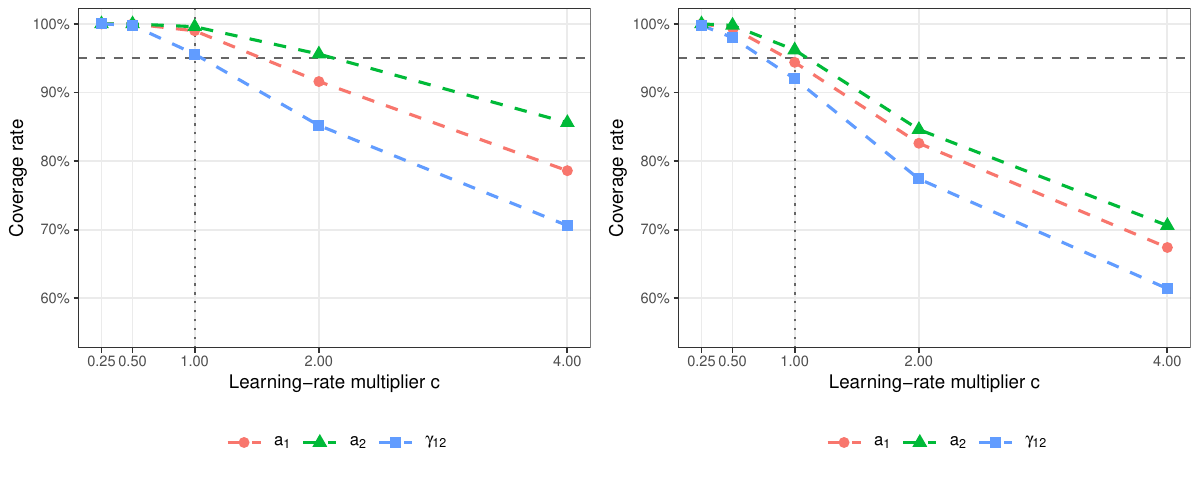}
    \caption{Sensitivity of DFD-Bayes coverage to the learning-rate multiplier $c$ in the misspecified Poisson graphical model at $n=1000$. Left: $\nu=0.8$. Right: $\nu=1.25$. Horizontal dashed lines indicate the nominal 95\% coverage level, and $c=1$ corresponds to the calibrated learning rate.}
    \label{fig:pg2_sensitivity}
\end{figure}

Across both misspecification settings, smaller learning rates produce conservative credible intervals, while larger learning rates lead to substantial undercoverage. The effect also varies across parameters; in particular, the interaction parameter $\gamma_{12}$ is generally more sensitive to increasing the learning rate. These results reinforce the sensitivity observed in the Potts experiment and illustrate that a single calibrated learning rate need not provide comparable uncertainty calibration across parameters.

\section{Additional Details for the Automobile Claim Data}
\label{app:cmp_insurance}

We provide additional computational details and diagnostics for the automobile claim analysis in Section~\ref{sec:real_data}. Because the data are available as grouped frequencies, the DFD moment calculations and Bayesian ETEL weights are evaluated using the distinct claim-count levels and their corresponding frequencies. This frequency-compressed implementation is algebraically equivalent to expanding the data into all $n=106{,}974$ individual observations.

For numerical stability, posterior sampling is performed in $(\log\lambda,\log\nu)$ coordinates, with the corresponding Jacobian included in the posterior density. Four independent Markov chains are run for 30,000 iterations, with the first 10,000 discarded as warm-up and every second post-warm-up draw retained, yielding 40,000 posterior draws across the four chains. The proposal scale is adapted during warm-up and held fixed thereafter. We assess convergence and sampling efficiency using rank-normalized $\widehat R$, bulk effective sample size (Bulk ESS), and tail effective sample size (Tail ESS) following
\citet{vehtari2021rank}.

\begin{table}[htbp]
\caption{MCMC diagnostics for the automobile claim analysis.}
\label{tab:cmp_insurance_diagnostics}
\centering
\small
\begin{tabular}{lrrr}
\toprule
Parameter & $\widehat R$ & Bulk ESS & Tail ESS \\
\midrule
$\lambda$ & 1.0002 & 9,620  & 12,221 \\
$\nu$     & 1.0005 & 10,139 & 11,651 \\
\bottomrule
\end{tabular}
\end{table}

The post-warm-up acceptance rates range from $0.283$ to $0.309$. Both parameters have $\widehat R<1.001$, bulk ESS above 9,600, and tail ESS above 11,600, indicating satisfactory convergence and sampling efficiency. Figure~\ref{fig:cmp_insurance_trace} shows the corresponding trace plots.

\begin{figure}[htbp]
\centering
\includegraphics[width=0.85\textwidth]
{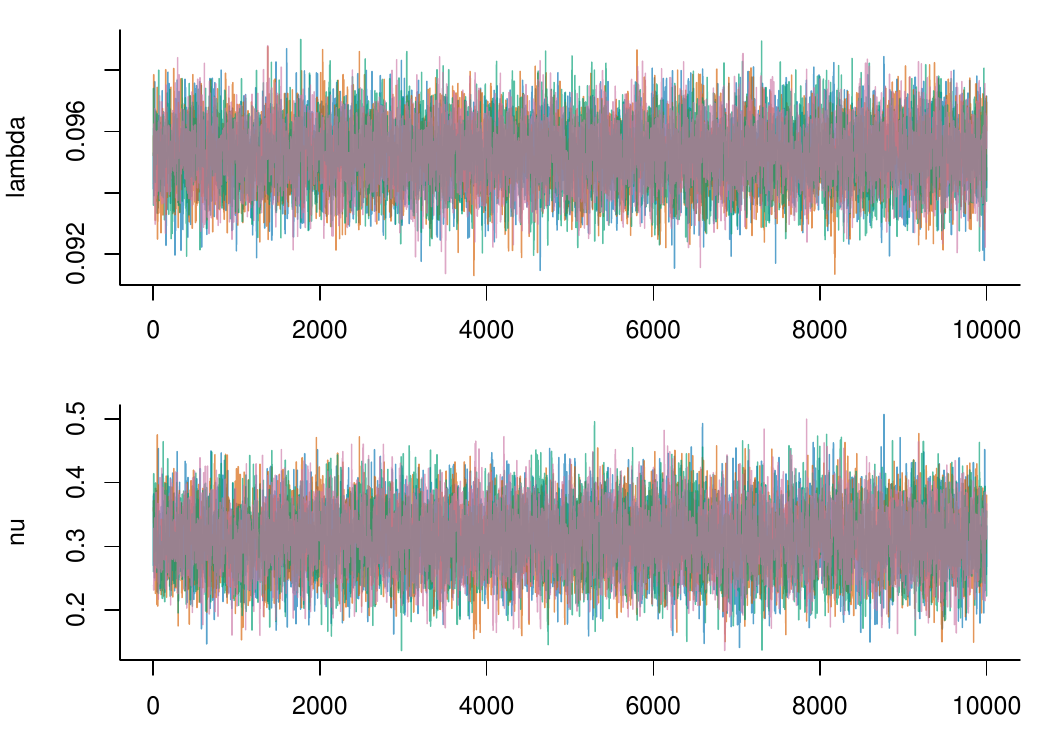}
\caption{Trace plots for $\lambda$ and $\nu$ from the four DFD-BETEL
chains in the automobile claim analysis.}
\label{fig:cmp_insurance_trace}
\end{figure}

\end{document}